\documentclass[12pt]{article}
\usepackage{amsmath}
\usepackage{graphicx,psfrag,epsf}
\usepackage{amssymb, amsmath, amsthm, color,enumerate,float}
\usepackage{algorithm, comment,subfig}
\usepackage[numbers]{natbib}
\usepackage{enumerate,multirow,booktabs}
\usepackage{url}
\usepackage{algpseudocode}
\usepackage{url} 

\newcommand{\blind}{0}
\newtheorem{lem}{Lemma}
\newtheorem{thm}{Theorem}
\newtheorem*{thm*}{Theorem}
\newtheorem{pro}{Proposition}
\newtheorem{defn}{Definition}
\newtheorem{rmk}{Remark}
\newtheorem{cor}{Corollary}
 \DeclareMathOperator{\var}{var}
\DeclareMathOperator{\tr}{tr} 
\DeclareMathOperator{\IMSE}{IMSE}

\DeclareMathOperator{\Eff}{Eff}
\DeclareMathOperator{\LEff}{LEff}
\DeclareMathOperator*{\argmin}{argmin}

\newcommand{\reals}{\mathbb{R}}

\newcommand{\vx}{\boldsymbol{x}}

\newcommand{\vc}{\boldsymbol{c}}

\newcommand{\vf}{\boldsymbol{f}}
\newcommand{\vg}{\boldsymbol{g}}

\newcommand{\vlambda}{\boldsymbol{\lambda}}

\newcommand{\vbeta}{\boldsymbol{\beta}}

\newcommand{\mA}{{\mathsf A}}
\newcommand{\mB}{{\mathsf B}}
\newcommand{\mC}{{\mathsf C}}

\newcommand{\mI}{{\mathsf I}}

\newcommand{\mS}{{\mathsf S}}
\newcommand{\E}{\mathbb{E}}

\newcommand{\dif}{{\rm d}}

\newcommand{\Ic}{\text{I}_\text{c}}
\newcommand{\EI}{\text{EI}}

\newcommand{\unif}{\text{\rm unif}}
\newcommand{\asine}{\text{\rm arcsine}}

\newcommand{\citeO}[1]{\bibpunct{}{}{;}{a}{}{,}\cite{#1}}
\newcommand{\citeI}[1]{\bibpunct{(}{)}{;}{a}{}{,}\cite{#1}}

\begin{document}

\def\spacingset#1{\renewcommand{\baselinestretch}%
{#1}\small\normalsize} \spacingset{1}


\if0\blind
{
  \title{\bf An Efficient Algorithm for Elastic I-optimal Design of Generalized Linear Models}
  \author{Yiou Li\thanks{Yiou Li is Assistant Professor, Department of Mathematical Sciences, DePaul University, Chicago, IL 60605 (E-mail: yli139@depaul.edu)}\hspace{.2cm}\\
    Department of Mathematical Sciences, DePaul University\\
    and \\
    Xinwei Deng \thanks{Xinwei Deng is Associate Professor, Department of Statistics, Virginia Tech, Blacksburg, VA 24061 (E-mail: xdeng@vt.edu)}\\
    Department of Statistics, Virginia Tech\\
    }
     \date{}
  \maketitle
} \fi

\if1\blind
{
  \title{\bf An Efficient Algorithm for Elastic I-optimal Design of Generalized Linear Models}
  \date{}
  \maketitle
} \fi

\begin{abstract}
The generalized linear models (GLMs) are widely used in  statistical analysis and the related design issues are undoubtedly challenging.
The state-of-the-art works mostly apply to design criteria on the estimates of regression coefficients.
The prediction accuracy is usually critical in modern decision making and artificial intelligence applications. It is of importance to study optimal designs from the prediction aspects for generalized linear models.
In this work, we consider the Elastic I-optimality as a prediction-oriented design criterion for generalized linear models, and develop efficient algorithms for such $\EI$-optimal designs.
By investigating theoretical properties for the optimal weights of any set of design points and extending the general equivalence theorem to the $\EI$-optimality for GLMs,
the proposed efficient algorithm adequately combines the Fedorov-Wynn algorithm and multiplicative algorithm.
It achieves great computational efficiency with guaranteed convergence.
Numerical examples are conducted to evaluate the feasibility and computational efficiency of the proposed algorithm.
\end{abstract}

\noindent%
{\it Keywords:} Fast convergence; Multiplicative algorithm; Prediction-oriented; Sequential algorithm
\vfill

\newpage
\spacingset{1.45} 
\section{Introduction}\label{section:intro}
The generalized linear model (GLM) is a flexible generalization of linear models by relating the response to the predictors through a link function (\citeO{72nelder}).
The GLMs are widely used in many statistical analyses with different applications including business analytics, image analysis, bioinformatics, etc.
From a data collection perspective, it is of great importance to understand the optimal designs for the generalized linear model, especially from the prediction aspects.

Suppose that an experiment has $d$ explanatory variables $\vx = [x_1,...,x_d]$, and let $\Omega_j$ be a measurable set of all possible levels of the $j$th explanatory variable.
Common examples of $\Omega_j$ are $[-1,1]$ and $\reals$.
The experimental region, $\Omega$, is some measurable subset of
$\Omega_1\times\cdots\times\Omega_d$. In a generalized linear model,
the response variable $Y(\vx)$ is assumed to follow a particular distribution in the exponential family, including normal, binomial, Poisson and gamma distribution, etc. A link function provides the relationship between the linear predictor, $\eta=\vbeta^T\vg(\vx)$, and the mean of the response $Y(\vx)$,
\begin{eqnarray}\label{eqn:GLM}
\mu(\vx) = \E[Y(\vx)] = h^{-1}\left(\vbeta^T\vg(\vx)\right),
\end{eqnarray}
where $\vg = [g_1,...,g_l]^T$ are the basis and $\vbeta=[\beta_1,\beta_2,...,\beta_{l}]^T$ are the corresponding regression coefficients.
Here $h: \reals\rightarrow\reals$ is the link function, and $h^{-1}$ is the inverse function of $h$.

In this work, we consider an approximate design $\xi$ as
$\xi = \left\{\begin{array}{ccc}
\vx_1,&...,&\vx_n\\
\lambda_1,&...,&\lambda_n
\end{array}\right\}$
with $\vx_i\neq \vx_j$ if $i\neq j$. The $\lambda_i (\lambda_i\geq 0)$ represents the fraction of experiments that is to be carried out at design point $\vx_i$ and $\sum\limits_{i=1}^n\lambda_i=1$.
The Fisher information matrix of the generalized linear model in \eqref{eqn:GLM} is:
\begin{equation}\label{eqn:fisher}
\mI(\xi, \vbeta) = \sum\limits_{i=1}^n\lambda_i\vg(\vx_i)w(\vx_i)\vg^T(\vx_i),
\end{equation}
where $w(\vx_i) = 1/\{\var(Y(\vx_i))[h^{'}(\mu(\vx_i))]^2\}$. The notation $\mI(\xi,\vbeta)$ emphasizes the dependence on the design $\xi$ and regression coefficient $\vbeta$.

For GLMs, the design issues tend to be much more challenging than those in linear models due to the complicated Fisher information matrix.
The Fisher information matrix $\mI(\xi, \vbeta)$ often depends on the regression coefficient $\vbeta$ through the function $w(\vx_i)$ in \eqref{eqn:fisher}.
Since most optimal design criteria can be expressed as a function of the Fisher information matrix,
a scientific understanding of locally optimal designs is often conducted under the assumption that some initial estimates of the model coefficients are available.
%
Most literature and algorithms on design of GLMs focus on the D-optimality or A-optimality for accurate estimation of regression coefficients.
\citeI{atkinson2015designs} provided a review of recent work on the designs of GLMs mainly based on D-optimality.
Practically, design criteria related to model prediction accuracy are of large interest, especially in decision-making and artificial intelligence (e.g., \citeO{07stein}; \citeO{settles2009active}; \citeO{bilgic2010active}).
Under such a consideration, the I-optimality that aims at minimizing the average variance of prediction over the experimental region $\Omega$ (\citeO{atkinson2014optimal}) is often used in the literature.
For the linear model,
\citeI{87haines} proposed a simulated annealing algorithm to obtain exact I-optimal design.
\citeI{88meyer, 95meyer} used simulated annealing and coordinate-exchange algorithms to construct exact I-optimal design.
However, there are few works on I-optimality for GLMs in the literature.

In practical applications such as additive manufacturing (\citeO{sun2018functional}), different regions of exploratory variables $\vx$ present different importance of interest.
For example, one may be interested in the prediction of response over a subregion  $\Omega_C$ of $\Omega$ instead of the whole experimental region $\Omega$ (\citeO{atkinson2014optimal}). For instance, the experimental region $\Omega$ could be the hypercube $[-1,1]^d$. But it is likely that the responses corresponding to positive explanatory variables are of more importance, or even only those responses are of interest. \citeI{06khu} proposed the criterion that minimizes the mean-squared error of the prediction at a single explanatory variable value, where $\Omega_C$ contains a single point.
Based on this motivation, we propose a so-called \emph{Elastic I-optimality} that is more general and flexible than the classical I-optimality.
The EI-optimality criterion aims at minimizing the integrated mean squared prediction error with respect to \textit{certain probability measure on the experimental region.}
This criterion shares a similar spirit but is more general than the criterion in the work of \citeI{box1963choice}.

The contribution of this work is to study a general and flexible prediction-oriented optimality criterion, EI-optimality, to fill the gap in the theory of EI-optimal designs for GLMs, and to advance an efficient algorithm of constructing EI-optimal designs for GLMs.
Specifically, we first establish the EI-optimality for GLMs and study theoretical properties of optimal fractions (i.e., weights) $\lambda_1, \ldots, \lambda_n$ given design points $\vx_1, \ldots, \vx_n$.
The resultant theoretical properties are not limited to EI-optimality, but also can be applied to other optimality criteria.
Based on these theoretical investigations, we develop an efficient sequential algorithm of constructing EI-optimal designs for GLMs.
The proposed algorithm sequentially adds points into design and optimizes their weights simultaneously using the multiplicative algorithm (\citeO{78titterington}).
The convergence properties of the proposed algorithm is established.
A key contribution of the proposed algorithm is to ensure the computational efficiency with sound theoretical convergence.
Our theoretical results provide theoretical insights on the optimal weights, which is very crucial to achieve a fast convergence rate for the proposed algorithm.
The advantages of the proposed algorithm are: (1) very easy to implement; (2) theoretically proved convergence; (3) computationally efficient; (4) suitable for other optimality criteria by simple modification.

The multiplicative algorithm, first proposed by \citeI{78titterington} and  \citeI{78silvey}, has been widely used  in finding optimal designs of linear regression models.
The main drawback of the multiplicative algorithm is that it requires a large candidate pool of points and updating the weights of all candidate points simultaneously can be computationally expensive.
However, in our proposed sequential algorithm, guided by the theoretical results on optimal weights, we can apply the multiplicative algorithm to a much smaller set of points,
which breaks the barrier of the original multiplicative algorithm and thus greatly improves the efficiency of the algorithm.
The proposed algorithm is not only computationally efficient, but also very simple to implement.
Furthermore, by employing the multiplicative algorithm, only nonnegative weights will be obtained and one needs not to deal with potential negative weights as in the algorithm proposed by \citeI{13yang}.
It is worth pointing out that the proposed algorithm can be easily extended to construct optimal design under other optimality criteria, like $\Phi_p$-optimality.

The remainder of this work is organized as follows.
The proposed prediction-oriented EI-optimality for GLMs is established in Section \ref{section:IMSE}.
Section \ref{section:optweights} details the findings on computing optimal weights given any set of design points.
The proposed sequential algorithm and its convergence properties are developed in Section \ref{sec: Algorithm}.
Numerical examples are conducted in Section \ref{section:Example} to evaluate the performance of the proposed algorithm.
We conclude this work with some discussion in Section \ref{sec: discussion}.
All the proofs are reported in the Appendix.



\section{The EI-Optimality Criterion for GLM}\label{section:IMSE}
The GLM is a generalization of various statistical models, including linear, logistic and Poisson regression.
For GLMs, making prediction of response $Y(\vx)$ at given input $\vx$ is always an important objective in many practical applications (e.g., \citeO{07stein}; \citeO{settles2009active}; \citeO{bilgic2010active}).
Thus, it is of great interest to adopt a prediction oriented criterion for computing the optimal design.
Following this idea, it is natural to consider the design based on mean response $\mu(\vx)$ that is square integrable with respect to some probability measure $\nu$ defined on $\mathbb{R}^d$. The associated probability distribution is given by $F_{\IMSE}(\vx)  = \nu\left(\prod\limits_{i=1}^d(-\infty,x_i]\right)$. Then, a general and flexible Elastic I-optimality criterion is defined as follows.
\begin{defn} \label{defn:IMSE}
The \textbf{elastic integrated mean squared prediction error (EIMSE)} is defined in terms of the difference between the true mean response, $\mu(\vx)$, and the fitted mean response $\hat{\mu}(\vx)$ as:
\begin{equation} \label{Res-IMSEdef}
\text{EIMSE}(\xi,\vbeta, F_{\IMSE})=\E\left[\int_{\Omega} \left(\hat{\mu}(\vx)-\mu(\vx)\right)^2\dif F_{\IMSE}(\vx) \right],
\end{equation}
where $F_{\IMSE}$ is the probability distribution induced by probability measure $\nu$.
\end{defn}

The classical I-optimality for linear models is defined as the average variance of response over the experimental region $\Omega$ (\citeO{atkinson2014optimal}):
$$I(\xi) = \int_{\Omega}\var(Y|\vx)\dif \vx\left/\int_{\Omega}\dif \vx\right. = \E\left[\int_{\Omega} \left(\hat{\mu}(\vx)-\mu(\vx)\right)^2\dif F_{\unif}(\vx) \right],$$
where $F_{\unif}$ is the uniform distribution on $\Omega$, $\hat{\mu}$ and $\mu$ are the fitted mean response and true mean response of the linear model, respectively. Obviously, the I-optimality $I(\xi)$ is a special case of EI-optimality in \eqref{Res-IMSEdef} with $\nu$ chosen to be the uniform probability measure.
In \citeI{06atk} (Chapter 10.6), the author briefly mentioned a similar criterion on average prediction variance involving a probability distribution, without
theoretical investigation and efficient algorithms.
Here, we formally propose the EI-optimality and discuss in detail about its advantages.
By introducing the probability measure $\nu$, the EI-optimatliy is more flexible to assist practical applications in several scenarios:
(1) the responses corresponding to $\vx$ on a subregion $\Omega_C$ are of particular interest; (2) the responses corresponding to different values of $\vx$ are of different importance; (3) the responses corresponding to finite number of $\vx$ values are of interest. It is worthy to point out that when $\nu$ is chosen to be the Dirac measure that puts a unit mass at $\vx_0$ that maximizes $\E\left[(\hat{\mu}(\vx)-\mu(\vx))^2\right]$, i.e., $\E\left[(\hat{\mu}(\vx_0)-\mu(\vx_0))^2\right] = \sup\limits_{\vx\in\Omega}\,\,\, \E\left[(\hat{\mu}(\vx)-\mu(\vx))^2\right]$, the  $\EI(\xi,\vbeta,F_{\IMSE})$ becomes the G-optimality (\citeO{06atk}) that focuses on the maximum variance of the mean response.

Under the context of generalized linear models, the fitted mean response $\hat{\mu}(\vx) = h^{-1}\left(\hat{\vbeta}^T\vg(\vx)\right)$ can be expanded around the true mean response $\mu(\vx) = h^{-1}\left(\vbeta^T\vg(\vx)\right)$ using Taylor expansion, which is
\begin{eqnarray*}
\hat{\mu}(\vx)-\mu(\vx) = h^{-1}\left(\hat{\vbeta}^T\vg(\vx)\right) - h^{-1}\left(\vbeta^T\vg(\vx)\right)\approx\vc(\vx)^T\left(\hat{\vbeta}-\vbeta\right),
\end{eqnarray*}
with $\boldsymbol{c}(\vx) = \left(\frac{\partial h^{-1}}{\partial\beta_1}(\vx),...,\frac{\partial h^{-1}}{\partial\beta_l}(\vx)\right)^T = \vg(\vx)\left(\frac{\dif h^{-1}}{\dif \eta}\right)$.
Here $\eta=\vbeta^T\vg(\vx)$ is the linear predictor. Then, using the above first-order Taylor expansion, the elastic integrated mean squared error defined in equation \eqref{Res-IMSEdef} could be approximated with,
$$\EI(\xi,\vbeta,F_{\IMSE}) = \E\left[\int_{\Omega} \vc(\vx)^T\left(\hat{\vbeta}-\vbeta\right)\left(\hat{\vbeta}-\vbeta\right)^T\vc(\vx)\dif  F_{\IMSE}\right].$$
In numerical analysis, the first-order Taylor expansion is often used to approximate the difference of a nonlinear function between adjacent points.
Considering the design issue for generalized linear models, similar approaches was commonly used in the literature, such as the work in \citeI{07stein} for the logistic regression models.

\begin{lem}\label{lem:IMSE}
For the GLMs in \eqref{eqn:GLM}, the Elastic IMSE $\EI(\xi,\vbeta,F_{\IMSE})$ can be expressed as
$$\EI(\xi,\vbeta,F_{\IMSE}) = \tr\left(\mA \mI(\xi,\vbeta)^{-1}\right),$$
where the matrix $\mA = \int_{\Omega}\vc(\vx)\vc(\vx)^T\dif F_{\IMSE}(\vx)=\int_{\Omega}\vg(\vx)\vg^T(\vx)\left[\frac{\dif h^{-1}}{\dif \eta}\right]^2\dif F_{\IMSE}(\vx)$ depends only on the regression coefficients, basis functions and the probability distribution $F_{\IMSE}$, but not the design. The Fisher information matrix $\mI(\xi,\vbeta)$, defined in equation (\ref{eqn:fisher}), depends on regression coefficients, basis functions and design, but not on the probability distribution $F_{\IMSE}$.
\end{lem}

Hereafter, we call the \textit{EI-optimality} as the corresponding optimality criterion aiming at minimizing
$$\EI(\xi,\vbeta,F_{\IMSE}) = \tr\left(\mA \mI(\xi,\vbeta)^{-1}\right).$$
A design $\xi^*$ is called an \textit{EI-optimal design} if it minimizes $\EI(\xi,\vbeta,F_{\IMSE})$. In this work, we will focus on local EI-optimal designs given some initial regression coefficient $\vbeta$.
To simplify the notation, $\vbeta$ will be omitted from the notation $\mI(\xi,\vbeta)$ of the Fisher information matrix, and $\mI(\xi)$ will be used instead.

\section{Algorithms for Finding Optimal Weights Given Design Points}\label{section:optweights}

In this section, we will investigate how to assign optimal weights to support points that minimize $\EI(\xi,\vbeta,F_{\IMSE})$ when the design points are given.
One popular method in the literature (\citeO{13yang}) is the Newton-Raphson technique to find the optimal weights, which calculates the roots of the first-order partial derivatives of the criterion with respect to the weights.
There are certain possible drawbacks of such a Newton-Raphson based method:
First, it may result in weights outside $[0,1]$ and thus further efforts are needed.
Second, it requires the inversion of Hessian matrix, which could be (numerically) singular.
Moreover, as noted in \citeI{13yang}, there is no guarantee of convergence.
The problems of Newton-Raphson method will be illustrated in the numerical examples in Section \ref{section:Example}.

Specifically, we will derive a theorem on optimal weights given design points for $\Phi_p$-optimality (\citeO{kiefer1974general}) in Section 3.1,
and then will take the mathematical structure of EI-optimality as a slight variation of $\Phi_1$-optimality.
Guided by this theorem, we develop an efficient algorithm (\textbf{Algorithm 1}) to find the optimal weights given design points in Section 3.2.
Note that \textbf{Algorithm 1} can be used for both $\Phi_p$-optimality and EI-optimality.
Interestingly, this algorithm coincides with the well-known multiplicative algorithm,  providing a good justification of applying multiplicative algorithm in our sequential algorithm in Section \ref{sec: Algorithm}.
There are several advantages of the multiplicative algorithm: simple to implement, ensuring feasible weights, guaranteed convergence, and no Hessian matrix inversion.

\subsection{Properties of Optimal Weights Given Design Points}
 Following the definition  given in \citeI{75kiefer}, the $\Phi_p$-optimality is defined as
$$\Phi_p(\xi) = \left[\tr\left(\mI(\xi)^{-p}\right)\right]^{1/p},\,\,\,0< p<\infty$$
with $\Phi_0(\xi)$ as D-optimality, $\Phi_{\infty}(\xi)$ as E-optimality and $\Phi_1(\xi)$ as A-optimality. Even more generally, one may be interested in several functions $\vf(\vbeta) = [f_1(\vbeta),...,f_q(\vbeta)]^T$ of regression coefficient $\vbeta$. Then, the more general $\Phi_p$-optimality  is defined as (\citeO{kiefer1974general})
\begin{eqnarray}\label{eq:phi_p}
\Phi_p(\xi) = \left(q^{-1}\tr\left[\frac{\partial \vf(\vbeta)}{\partial \vbeta^T}\mI(\xi)^{-1}\left(\frac{\partial \vf(\vbeta)}{\partial \vbeta^T}\right)^T\right]^p\right)^{1/p},\,\,\,0<p<\infty.
\end{eqnarray}

Given the fixed design points $\vx_1,\vx_2,...,\vx_n$,
define $\vlambda = [\lambda_1,\lambda_2,...,\lambda_n]^T$ to be the weight vector with $\lambda_i$ as the weight of the corresponding design point $\vx_i$.
We write the corresponding design as
$$\xi^{\vlambda} = \left\{\begin{array}{ccc}
\vx_1,&...,&\vx_n\\
\lambda_1,&...,&\lambda_n
\end{array}\right\}.$$
A superscript $\vlambda$ is added to emphasize that only the weight vector is changeable in the design under this situation. The optimal weight vector $\vlambda^*$ should be the one that minimizes $\Phi_p(\xi^{\vlambda})$ in \eqref{eq:phi_p} with design points $\vx_1,...,\vx_n$ fixed.

\begin{lem}\label{lem:ConWeight}
The $\Phi_p$ is a convex function of weight vector $\vlambda$, and
$$\Phi_p(\xi^{\vlambda}) = \left[q^{-1}\tr\left[\mI^{-1}(\xi^{\vlambda})\mB\right]^{p}\right]^{1/p},$$
where $\mB = \left(\frac{\partial \vf(\vbeta)}{\partial \vbeta^T}\right)^T\left(\frac{\partial \vf(\vbeta)}{\partial \vbeta^T}\right)$ is positive semidefinite with size $l\times l$ and rank $q\leq l$, with $q$ as the length of $\vf$.
\end{lem}

Consider another weight vector $\Delta{\vlambda} = [\Delta{\lambda}_1,...,\Delta{\lambda}_n]^T$ with $0\leq \Delta{\lambda}_i\leq 1, i=1,...,n$ and $\sum\limits_{i=1}^n\Delta{\lambda}_i = 1$ .
Then the convex combination $\tilde{\vlambda} = (1-\alpha)\vlambda+\alpha\Delta{\vlambda}$ with $0\leq\alpha\leq 1$ is also a feasible weight vector.
Define the directional derivative of $\Phi_p(\xi^{\vlambda})$ in the direction of a new weight vector $\Delta{\vlambda} = [\Delta{\lambda}_1,...,\Delta{\lambda}_n]$ as:
$$\psi(\Delta{\vlambda},\vlambda) = \lim\limits_{\alpha\rightarrow 0^+}\frac{\Phi_p(\xi^{\tilde{\vlambda}})-\Phi_p(\xi^{\vlambda})}{\alpha}.$$

\begin{lem}\label{lem:DerWeight}
The directional derivative of $\Phi_p(\xi^{\vlambda})$ in the direction of a new weight vector $\Delta{\vlambda} = [\Delta{\lambda}_1,...,\Delta{\lambda}_n]^T$ can be calculated as,
$$\psi(\Delta{\vlambda},\vlambda) = \Phi_p(\xi^{\vlambda})-q^{-1/p}\left[\tr\left(\mI(\xi^{\vlambda})^{-1}\mB\right)^p\right]^{1/p-1}\tr\left[\left(\mI(\xi^{\vlambda})^{-1}\mB\right)^{p-1} \mI(\xi^{\vlambda})^{-1}\mI(\xi^{\Delta\vlambda})\mI(\xi^{\vlambda})^{-1}\mB\right],$$
where $\mB = \left(\frac{\partial \vf(\vbeta)}{\partial \vbeta^T}\right)^T\left(\frac{\partial \vf(\vbeta)}{\partial \vbeta^T}\right)$ .
\end{lem}

\begin{rmk}\label{rmk:Ic}
Note that $\mA = \int_{\Omega}\vg(\vx)\vg^T(\vx)\left[\frac{\dif h^{-1}}{\dif \eta}\right]^2\dif F_{\IMSE}(\vx)$ is an $l\times l$ positive semi-definite matrix. Although EI-optimality is not a member of $\Phi_p$-optimality, since it has the same mathematical structure as $\Phi_1(\xi)$ with $\mB = q\mA$ defined in Lemma \ref{lem:IMSE}, the mathematical properties of $\Phi_p$-optimality can be applied to EI-optimality.
\end{rmk}

\begin{cor} \label{cor:I-weight_der}
For EI-optimality, the directional derivative in the direction of a new weight vector $\Delta{\vlambda}$ is
$$\psi(\Delta{\vlambda},\vlambda) = \tr\left(\mI(\xi^{\vlambda})^{-1}\mA\right)-\sum\limits_{i=1}^n\Delta{\lambda}_iw(\vx_i)\vg(\vx_i)^T\mI(\xi^{\vlambda})^{-1}\mA\mI(\xi^{\vlambda})^{-1}\vg(\vx_i).$$
\end{cor}

The following theorem provides a necessary and sufficient condition of optimal weight vector that minimizes $\Phi_p(\xi)$ when the design points are fixed.
\begin{thm}\label{thm:OptWeights}
Given a fixed set of design points $\vx_1,...,\vx_n$, the weight vector $\vlambda^* = [\lambda^*_1,...,\lambda^*_n]^T$ minimizes $\Phi_p(\xi^{\vlambda})$ if and only if,
$$\Phi_p(\xi^{\vlambda^*})=q^{-1/p}\left[\tr\left(\mI(\xi^{\vlambda^*})^{-1}\mB\right)^p\right]^{1/p-1}w(\vx_i)\vg(\vx_i)^T\mI(\xi^{\vlambda^*})^{-1}\mB\left(\mI(\xi^{\vlambda^*})^{-1}\mB\right)^{p-1}\mI(\xi^{\vlambda^*})^{-1}\vg(\vx_i),$$ for all design points $\vx_i$ with $\lambda_i^*>0$;
and $$\Phi_p(\xi^{\vlambda^*})\geq q^{-1/p}\left[\tr\left(\mI(\xi^{\vlambda^*})^{-1}\mB\right)^p\right]^{1/p-1}w(\vx_j)\vg(\vx_j)^T\mI(\xi^{\vlambda^*})^{-1}\mB\left(\mI(\xi^{\vlambda^*})^{-1}\mB\right)^{p-1}\mI(\xi^{\vlambda^*})^{-1}\vg(\vx_j),$$ for all design points $\vx_j$ with $\lambda_j^*=0$.
\end{thm}

Note that in Theorem \ref{thm:OptWeights}, for an optimal weight vector, it is required to hold the equality for the nonzero weights.
The following results provides sufficient conditions that a weight vector minimizes $\Phi_p(\xi^{\vlambda})$ and $\EI(\xi^{\vlambda},\vbeta,F_{\IMSE})$, respectively.

\begin{cor}\label{cor:suf_weight}
(i.) If a weight vector $\vlambda^* = [\lambda^*_1,...,\lambda^*_n]^T$ satisfies
$$\Phi_p(\xi^{\vlambda^*})=q^{-1/p}\left[\tr\left(\mI(\xi^{\vlambda^*})^{-1}\mB\right)^p\right]^{1/p-1}w(\vx_i)\vg(\vx_i)^T\mI(\xi^{\vlambda^*})^{-1}\mB\left(\mI(\xi^{\vlambda^*})^{-1}\mB\right)^{p-1}\mI(\xi^{\vlambda^*})^{-1}\vg(\vx_i),$$ for all design points $\vx_1,...,\vx_n$, then $\vlambda^*$ minimizes $\Phi_p(\xi^{\vlambda})$. \\
(ii.) For EI-optimality, a sufficient condition that $\vlambda^*$ minimizes $EI(\xi^{\vlambda},\vbeta,F_{\IMSE})$ is
$$\tr(\mI(\xi^{\vlambda^*})^{-1}\mA) = w(\vx_i)\vg(\vx_i)^T\mI(\xi^{\vlambda^*})^{-1}\mA\mI(\xi^{\vlambda^*})^{-1}\vg(\vx_i),\,\,\,\text{for},\,\,\,\, i=1,...,n.$$
\end{cor}


The results in Theorem \ref{thm:OptWeights} and Corollary \ref{cor:suf_weight} provides a useful gateway to design an effective algorithm for finding the optimal weights given the design points.

\subsection{Multiplicative Algorithm for Optimal Weights}\label{subsec:mulalg}
As shown in Corollary \ref{cor:suf_weight}, the solution of a system of equations
\begin{equation}\label{eqn:OptWeights}
\Phi_p(\xi^{\vlambda^*})=q^{-1/p}\left[\tr\left(\mI(\xi^{\vlambda^*})^{-1}\mB\right)^p\right]^{1/p-1}w(\vx_i)\vg(\vx_i)^T\mI(\xi^{\vlambda^*})^{-1}\mB\left(\mI(\xi^{\vlambda^*})^{-1}\mB\right)^{p-1}\mI(\xi^{\vlambda^*})^{-1}\vg(\vx_i),
\end{equation}
$i=1,...,n$, is a set of optimal weights that minimizes $\Phi_p(\xi^{\vlambda})$.

The weight of a design point $\vx_i$ should be adjusted according to the values of the two sides of equation (\ref{eqn:OptWeights}). Note that for any weight vector $\vlambda = [\lambda_1,...,\lambda_n]$, we have
$$\Phi_p(\xi^{\vlambda}) = \sum\limits_{i=1}^n\lambda_iq^{-1/p}\left[\tr\left(\mI(\xi^{\vlambda})^{-1}\mB\right)^p\right]^{1/p-1}w(\vx_i)\vg(\vx_i)^T\mI(\xi^{\vlambda})^{-1}\mB\left(\mI(\xi^{\vlambda})^{-1}\mB\right)^{p-1}\mI(\xi^{\vlambda})^{-1}\vg(\vx_i).$$
Based on the above observation, the weight $\lambda_i$ of a design point $\vx_i$ should be adjusted according to the value of
$q^{-1/p}\left[\tr\left(\mI(\xi^{\vlambda})^{-1}\mB\right)^p\right]^{1/p-1}w(\vx_i)\vg(\vx_i)^T\mI(\xi^{\vlambda})^{-1}\mB\left(\mI(\xi^{\vlambda})^{-1}\mB\right)^{p-1}\\\mI(\xi^{\vlambda})^{-1}\vg(\vx_i).$ Our strategy of finding optimal weights would be: if
$$\Phi_p(\xi^{\vlambda}) < q^{-1/p}\left[\tr\left(\mI(\xi^{\vlambda})^{-1}\mB\right)^p\right]^{1/p-1}w(\vx_i)\\\vg(\vx_i)^T\mI(\xi^{\vlambda})^{-1}\mB\left(\mI(\xi^{\vlambda})^{-1}\mB\right)^{p-1}\mI(\xi^{\vlambda})^{-1}\vg(\vx_i),$$ the weight $\lambda_i$ of $\vx_i$ should be increased. On the other hand, if
$$\Phi_p(\xi^{\vlambda}) > q^{-1/p}\big[\tr\left(\mI(\xi^{\vlambda})^{-1}\mB\right)^p\big]^{1/p-1}w(\vx_i)\vg(\vx_i)^T\mI(\xi^{\vlambda})^{-1}\mB\left(\mI(\xi^{\vlambda})^{-1}\mB\right)^{p-1}\mI(\xi^{\vlambda})^{-1}\vg(\vx_i),$$ the weight $\lambda_i$ of $\vx_i$ should be decreased. Thus, the ratio
$$\frac{q^{-1/p}\left[\tr\left(\mI(\xi^{\vlambda})^{-1}\mB\right)^p\right]^{1/p-1}w(\vx_i)\vg(\vx_i)^T\mI(\xi^{\vlambda})^{-1}\mB\left(\mI(\xi^{\vlambda})^{-1}\mB\right)^{p-1}\mI(\xi^{\vlambda})^{-1}\vg(\vx_i)}{\Phi_p(\xi^{\vlambda})}$$ would indicate a good adjustment for the current weight of a design point $\vx_i$.

The detailed algorithm is described in \textbf{Algorithm 1} as follows.

\begin{algorithm}[H]\label{alg-1}
  \caption{}
  \begin{algorithmic}[1]
  \State Assign a random weight vector $\vlambda^0 = [\lambda_1^0,...,\lambda_n^0]^T$, and $k=0$.
  \While {$change>1e-15$ and $k<maxrun$}
  \For {$i = 1, \ldots, n$}
 \State Update the weight of design point $\vx_i$:
\begin{eqnarray}\label{for:multialg}
\lambda_i^{k+1} &=& \lambda_i^k \frac{\left(\frac{ q^{-1/p}\left[\tr\left(\mI(\xi^{\vlambda^k})^{-1}\mB\right)^p\right]^{1/p-1}w(\vx_i)\vg(\vx_i)^T\mI(\xi^{\vlambda^k})^{-1}\mB\left(\mI(\xi^{\vlambda^k})^{-1}\mB\right)^{p-1}\mI(\xi^{\vlambda^k})^{-1}\vg(\vx_i)}{ \Phi_p(\xi^{\vlambda})}\right)^\delta}{\sum\limits_{i=1}^n \lambda_i^k\left(\frac{q^{-1/p}\left[\tr\left(\mI(\xi^{\vlambda^k})^{-1}\mB\right)^p\right]^{1/p-1}w(\vx_i)\vg(\vx_i)^T\mI(\xi^{\vlambda^k})^{-1}\mB\left(\mI(\xi^{\vlambda^k})^{-1}\mB\right)^{p-1}\mI(\xi^{\vlambda^k})^{-1}\vg(\vx_i)}{ \Phi_p(\xi^{\vlambda})}\right)^\delta}\nonumber\\
& =& \lambda_i^k \frac{\left[w(\vx_i)\vg(\vx_i)^T\mI(\xi^{\vlambda^k})^{-1}\mB\left(\mI(\xi^{\vlambda^k})^{-1}\mB\right)^{p-1}\mI(\xi^{\vlambda^k})^{-1}\vg(\vx_i)\right]^\delta}{\sum\limits_{i=1}^n \lambda_i^k\left[ w(\vx_i)\vg(\vx_i)^T\mI(\xi^{\vlambda^k})^{-1}\mB\left(\mI(\xi^{\vlambda^k})^{-1}\mB\right)^{p-1}\mI(\xi^{\vlambda^k})^{-1}\vg(\vx_i)\right]^\delta},\footnotemark
\end{eqnarray}
\State $change=\max\limits_{i=1,\cdots,n}(|\lambda_i^{k+1}-\lambda_i^{k}|)$.
\State $k=k+1$.
\EndFor
 \EndWhile
\end{algorithmic}
\end{algorithm}
\footnotetext{The formula applies to D-optimality with $p=0$.}

We would like to remark that the algorithm turns out to be the well-known multiplicative algorithm proposed by \citeI{78titterington} and \citeI{78silvey}. Originally, a heuristic explanation for multiplicative algorithm is that $\lambda_i^{k+1} \propto \lambda_i^{k}\left(\left.\frac{\partial \Phi_p(\xi^{\lambda})}{\partial \lambda_i}\right|_{\vlambda = \vlambda^k}\right)^{\delta}$.
Here, we obtain the same algorithm based on a sufficient condition of optimal weights given in Corollary \ref{cor:suf_weight}. This may also explain why the multiplicative algorithm tends to converge slowly and result in many support points when there is a large candidate set. When the multiplicative algorithm is applied to a large candidate set, a very strong sufficient but not necessary condition is imposed on all the candidate points. This condition should be only imposed on the points with non-zero weights. However, if the design points are appropriately chosen so that most of them have nonzero weight, then the sufficient condition in Corollary \ref{cor:suf_weight} becomes almost necessary condition for the optimal weights.

In \textbf{Algorithm 1}, there are two user chosen parameters. One is convergence parameter $\delta\in (0,1)$ and the other one $maxrun$ is the maximum number of iterations allowed. Following the suggestion by \citeI{74fellman} and \citeI{83tor}, convergence parameter $\delta=\frac{1}{2}$ is usually used for A-optimality. Since EI-optimality has similar mathematical structure to A-optimality, we choose $\delta=\frac{1}{2}$ in our algorithm. It is also observed that the computational performance of the proposed algorithm is regarding $\delta$ in \emph{Example 2} of Section \ref{section:Example}.  We will show later in proposition \ref{pro:convAlg1} that \textbf{Algorithm 1} converges to optimal weights monotonically. Thus, the algorithm is not very sensitive to the choice of maximum iterations allowed, $maxrun$. In all our numerical examples, we choose $maxrun=100$, while the same parameter in \citeI{13yang} for Newton's method is 40.
The following theorem provides a sufficient condition when the iterative formula in \textbf{Algorithm 1} is feasible, and Corollary \ref{cor: feasibility} shows that, for EI-optimality, one can ensure that every iteration in \textbf{Algorithm 1} is always feasible by choosing appropriate basis functions.
\begin{thm}\label{thm:feasibility}
Suppose $\mI(\xi^{\vlambda^k})^{-1}\mB\neq 0$ in iteration $k$, then  the iteration in equation (\ref{for:multialg}) is feasible, i.e., the denominator
$$\sum\limits_{i=1}^n \lambda_i^k\left[w(\vx_i)\vg(\vx_i)^T\mI(\xi^{\vlambda^k})^{-1}\mB\left(\mI(\xi^{\vlambda^k})^{-1}\mB\right)^{p-1}\mI(\xi^{\vlambda^k})^{-1}\vg(\vx_i)\right]^\delta$$
is positive.
\end{thm}

For the EI-optimality, the iterative formula in \textbf{Algorithm 1} becomes
\begin{eqnarray}\label{for:multialgIc}
\lambda_i^{k+1} &=&  \lambda_i^k \frac{\left[w(\vx_i)\vg(\vx_i)^T\mI(\xi^{\vlambda^k})^{-1}\mA\mI(\xi^{\vlambda^k})^{-1}\vg(\vx_i)\right]^\delta}{\sum\limits_{i=1}^n \lambda_i^k\left[ w(\vx_i)\vg(\vx_i)^T\mI(\xi^{\vlambda^k})^{-1}\mA\mI(\xi^{\vlambda^k})^{-1}\vg(\vx_i)\right]^\delta},
\end{eqnarray}
with $\mA = \int_{\Omega}\vg(\vx)\vg^T(\vx)\left[\frac{\dif h^{-1}}{\dif \eta}\right]^2\dif F_{\IMSE}(\vx)$.

\begin{cor}\label{cor: feasibility}
For the EI-optimality, a positive definite $\mA = \int_{\Omega}\vg(\vx)\vg^T(\vx)\left[\frac{\dif h^{-1}}{\dif \eta}\right]^2\dif F_{\IMSE}(\vx)$ would ensure all the iterations in equation (\ref{for:multialg}) to be always feasible, i.e., the denominator
$$\sum\limits_{i=1}^n \lambda_i^k\left[w(\vx_i)\vg^T(\vx_i)\mI(\xi^{\vlambda^k})^{-1}\mA\mI(\xi^{\vlambda^k})^{-1}\vg(\vx_i)\right]^\delta$$
is always positive.
\end{cor}


For the $\Phi_p$-optimality, the matrix $\mB = \left(\frac{\partial \vf(\vbeta)}{\partial \vbeta^T}\right)^T\left(\frac{\partial \vf(\vbeta)}{\partial \vbeta^T}\right)$ is determined by $\vf(\vbeta)$, the functions of $\vbeta$ that are of interest, which are determined by the purpose of the experiment. But, for EI-optimality, the matrix $\mA = \int_{\Omega}\vg(\vx)\vg^T(\vx)\left[\frac{\dif h^{-1}}{\dif \eta}\right]^2\dif F_{\IMSE}(\vx)$  is always positive semi-definite and pre-determined before the optimal design is constructed.
When a singular $\mA$ is observed in the first step, one can choose the basis functions $g_1,...,g_l$ carefully so that they are (nearly) orthogonal functions such that
$$\int_{\Omega} g_i(\vx)g_j(\vx)\left[\frac{\dif h^{-1}}{\dif \eta}\right]^2\dif F_{\IMSE}(\vx) \approx 0\,\,\,\,\text{for}\,\,\,\,i\neq j.$$
By doing this, one can ensure the matrix $\mA$ to be positive definite.
In general, the Gram-Schmidt orthogonalization (\citeO{pur91}) can be used to construct uni/multi-variate orthogonal basis functions.

We can provide the convergence property of Algorithm 1, which generalizes Yu's work (\citeO{10yu}) to a broader class of optimality criteria, $\Phi_p$-optimality.
\begin{pro}[\textbf{Convergence of Algorithm 1}]\label{pro:convAlg1}
Given the design points $\vx_1,...,\vx_n$ fixed, the weight vector $\vlambda^k$ obtained from Algorithm 1 monotonically converges to the optimal weight vector $\vlambda^*$ that minimizes $\Phi_p$-optimality, as $k\rightarrow \infty$.
\end{pro}

\section{Proposed Sequential Algorithm of Constructing EI-Optimal Design}\label{sec: Algorithm}
In this section, we will describe the proposed efficient algorithm for constructing EI-optimal design for GLMs. General equivalence theorem is one of the main theoretical tools to develop algorithms to construct optimal designs.  Many Wynn-Federov type algorithms are developed based on general equivalence theorem (See \citeO{70wynn, 72wynn, 73whi, 13yang, martin2015combined}). \citeI{yang2016optimal} developed an efficient algorithm to construct $2^k$ D-optimal factorial design with binary response based on a specialized version of general equivalence theorem on a pre-determined finite set of design points.
We will first establish the general equivalence theorem for EI-optimality of GLMs in Section \ref{section:GET}, which provides an intuitive way for choosing the support points to construct an EI-optimal design in a sequential fashion.  Section \ref{section: Algorithm 2} details the proposed algorithm and develops the convergence property of the proposed algorithm.

\subsection{General Equivalence Theorem of EI-Optimality}\label{section:GET}
As shown in Remark \ref{rmk:Ic}, the EI-optimality has a similar mathematical structure as the $\Phi_1$-optimatliy, but the two criteria have different practical interpretations.
The General Equivalence Theorem of $\Phi_p$-optimality for generalized linear models has been established (\citeO{stufken2011optimal}), and it could be extended to EI-optimality for GLMs easily. We state the General Equivalence Theorem of EI-optimality for GLMs in the following Theorem \ref{thm:GET}, and the standard proof is omitted.
The extended theoretical results facilitate the sequential choice of support points as that of Fedorov-Wynn algorithm (\citeO{70wynn,72fed}).

Given two designs $\xi$ and $\xi'$, let the design $\tilde{\xi}$ be constructed as
$$\tilde{\xi} = (1-\alpha)\xi+\alpha\xi'.$$
Then, the derivative of $\EI(\xi,\vbeta,F_{\IMSE})$ in the direction of design $\xi'$ is
\begin{equation}\label{eqn:dirder}
\phi(\xi',\xi) = \lim\limits_{\alpha\rightarrow 0^+}\frac{\EI(\tilde{\xi},\vbeta,F_{\IMSE})-\EI(\xi,\vbeta,F_{\IMSE})}{\alpha}.
\end{equation}

\begin{lem}\label{lem:TrConvex}
The $\EI(\cdot,\vbeta,F_{\IMSE})$ is a convex function of the design $\xi$.
\end{lem}

With some algebra, we can obtain the directional derivative of $EI(\xi,\vbeta,\Omega_c)$ in the direction of any design $\xi'$ as,
$$\phi(\xi',\xi) = \lim\limits_{\alpha\rightarrow 0^+}\frac{\EI(\tilde{\xi},\vbeta,F_{\IMSE})-\EI(\xi,\vbeta,F_{\IMSE})}{\alpha} = \tr(\mI(\xi)^{-1}\mA)-\tr(\mI(\xi')\mI(\xi)^{-1}\mA\mI(\xi)^{-1}),$$
where $\mA =\int_{\Omega}\vg(\vx)\vg^T(\vx)\left[\frac{\dif h^{-1}}{\dif \eta}\right]^2\dif F_{\IMSE}(\vx)$.
Moreover, we can also get
the directional derivative of $\EI(\xi,\vbeta,F_{\IMSE})$ in the direction of a single point $\vx$ as,
$$\phi(\vx,\xi) = \tr(\mI(\xi)^{-1}\mA)-w(\vx)\vg(\vx)^T\mI(\xi)^{-1}\mA\mI(\xi)^{-1}\vg(\vx).$$


\begin{thm}[\textbf{General Equivalence Theorem}]\label{thm:GET}
The following three conditions of $\xi^*$ are equivalent:
\begin{enumerate}
\item
The design $\xi^{*}$ minimizes
$$\EI(\xi,\vbeta,\Omega_c) = \tr(\mA\mI(\xi)^{-1}).$$
\item
The design $\xi^*$ minimizes
$$\sup\limits_{\vx\in\Omega} w(\vx)\vg^T(\vx)\mI(\xi)^{-1}\mA\mI(\xi)^{-1}\vg(\vx).$$

\item
$w(\vx)\vg^T(\vx)\mI(\xi^*)^{-1}\mA\mI(\xi^*)^{-1}\vg(\vx)\leq \tr(\mI(\xi^*)^{-1}\mA)$ holds over the experimental region $\Omega$, and the equality holds only at the points of support of the design $\xi^*$.
\end{enumerate}
\end{thm}

According to Theorem \ref{thm:GET}, for an optimal design $\xi^*$, the directional derivative $\phi(\vx,\xi^*)$ is non-negative for any $\vx\in \Omega$.
It implies that for any non-optimal design, there will be some directions in which the directional derivative $\phi(\vx,\xi)<0$.
Given a current design $\xi$, to gain the maximal decrease in the EI-optimality criterion,
we would choose a new support point $\vx^*$ to be added into the design, if $\phi(\vx^*,\xi) = \min\limits_{\vx\in\Omega} \phi(\vx,\xi)<0$.
Then one can optimize the weights of all support points in the updated design, which is described in \textbf{Algorithm 1} in Section \ref{section:optweights}. With this greedy search of the design points, we hope that most/all optimal weights in each iteration are nonzero, and the multiplicative algorithm converges quickly.
By iterating the selection of support point and the weight update of all support points, this two-step iterative procedure
can be continued until $\min\limits_{\vx\in\Omega} \phi(\vx,\xi)\geq 0$ for all $\vx\in\Omega$, which means the updated design is EI-optimal.
Such a sequential algorithm of constructing optimal designs, as described in \textbf{Algorithm 2} in Section \ref{section: Algorithm 2},
follows similar spirits in the widely-used Fedorov-Wynn algorithm (\citeO{70wynn};\citeO{72wynn}), the multi-stage algorithm proposed by \citeI{13yang}, and the combined algorithm proposed by \citeI{martin2015combined}. It is worth pointing out that when the region of explanatory variable $\vx$ we are interested in, $\Omega_C$ is a subset of original experimental region $\Omega$, the optimal design $\xi^*$ is still defined and searched on the original experimental region $\Omega$. As a result, the support points in the optimal design may locate outside $\Omega_C$.

\subsection{The Proposed Sequential Algorithm}\label{section: Algorithm 2}
As discussed in Section 3.2, the multiplicative algorithm tends to converge slowly under a large candidate set.
The theoretical results in Theorem \ref{thm:GET} provide insightful guidelines on sequential selection of design points.
In combination with Algorithm 1 to find optimal weights for fixed design points in Section \ref{section:optweights},
we propose an efficient sequential algorithm to construct the EI-optimal design for generalized linear models.
The details of the proposed sequential algorithm is summarized in \textbf{Algorithm 2}.

\begin{algorithm}[H]\label{alg:alg2}
  \caption{}
  \begin{algorithmic}[1]
  \State Calculate matrix $\mA = \int_{\Omega}\vg(\vx)\vg^T(\vx)\left[\frac{\dif h^{-1}}{\dif \eta}\right]^2\dif F_{\IMSE}(\vx)$.
  \If{cond($\mA>1e16$)}
  \State Construct orthogonal basis $\vg$ using Gram-Schmidt orthogonalization.
  \State Calculate matrix $\mA$ using the new orthogonal basis $\vg$.
  \EndIf
  \State Generate an $N$ points candidate pool $\mathcal{C}$ using grid or Sobol sequence from experimental region $\Omega$.
  \State Choose an initial design points set $\mathcal{X}^{0} = \left\{\vx_1,\cdots,\vx_{l+1}\right\}$ containing $l+1$ points.
  \State Obtain optimal weights $\vlambda^0$ of initial design points set $\mathcal{X}^{0}$ using {\bf Algorithm 1} and form the initial design $\xi^0 = \left\{\begin{array}{cc} \mathcal{X}^{0}\\ \vlambda^0
  \end{array}\right\}$.
  \State Calculate the lower bound of EI-efficiency of $\xi^0$, $$\LEff(\xi^0|\xi^*) = \frac{\tr(\mI^{-1}(\xi^0)\mA)}{\max\limits_{\vx\in\mathcal{C}}w(\vx)\vg(\vx)^T\mI(\xi^0)^{-1}\mA\mI(\xi^0)^{-1}\vg(\vx)}.$$
  \State Set $r=1$.
  \While {$\LEff_{\text{EI}}(\xi^{r-1}|\xi^*)<reqeff$ and $r<maxiter$}
  \State
 Add the point $\vx_r^* = \argmin \limits_{\vx \in \mathcal{C}} \phi(\vx,\xi^{r-1})$ to the current design points set, i.e., $\mathcal{X}^{r} = \mathcal{X}^{r-1}\cup \vx_r^*$,
  where $\phi(\vx,\xi^{r-1})$ is the directional derivative expressed as
  $$\phi(\vx,\xi^{r-1}) = \tr(\mI(\xi^{r-1})^{-1}\mA)-w(\vx)\vg(\vx)^T\mI(\xi^{r-1})^{-1}\mA\mI(\xi^{r-1})^{-1}\vg(\vx). $$
\State Obtain optimal weights $\vlambda^r$ of the current design points set $\mathcal{X}^{r}$ using {\bf Algorithm 1} and form the current design $\xi^r = \left\{\begin{array}{cc} \mathcal{X}^{r}\\ \vlambda^r
  \end{array}\right\}$.
  \State Calculate the lower bound of EI-efficiency of $\xi^r$, $$\LEff_{\text{EI}}(\xi^r|\xi^*) = \frac{\tr(\mI^{-1}(\xi^r)\mA)}{\max\limits_{\vx\in\mathcal{C}}w(\vx)\vg(\vx)^T\mI(\xi^r)^{-1}\mA\mI(\xi^r)^{-1}\vg(\vx)}.$$
\State $r=r+1$.
 \EndWhile
\end{algorithmic}
\end{algorithm}

The efficiency of a design $\xi$ relative to another design $\xi'$ under a $\Phi_p$-optimality is defined as (\citeO{pukelsheim2006optimal})
$$\Eff_{\Phi_p}(\xi|\xi') = \frac{\Phi_p(\xi')}{\Phi_p(\xi)}.$$
Asymptotically, the efficiency is essentially  the ratio between the number of trials required to obtain the same amount of information measured by a specific $\Phi_p$-optimality using design $\xi'$ and $\xi$.
The following proposition provides a lower bound $\LEff_{\Phi_p}(\xi|\xi^*)$ of $\Phi_p$-efficiency of a design $\xi$ relative to the corresponding $\Phi_p$-optimal design $\xi^*$, and this lower bound is used in the stopping rule in \textbf{Algorithm 2}.

\begin{thm}\label{thm:efflowerbound}
With candidate pool $\mathcal{C}$, the $\Phi_p$-efficiency of any design $\xi$ relative to $\Phi_p$-optimal design $\xi^*$, $\Eff_{\Phi_p}(\xi|
\xi^*)$, is lower bounded by
\begin{eqnarray}\label{eqn:efflowerbound}
\LEff_{\Phi_p}(\xi|\xi^*)  = \frac{\Phi_p(\xi)}{\max\limits_{\vx\in\mathcal{C}}q^{-1/p}\left[\tr\left(\mI(\xi)^{-1}\mB\right)^p\right]^{1/p-1}w(\vx_i)\vg(\vx_i)^T\mI(\xi)^{-1}\mB\left(\mI(\xi)^{-1}\mB\right)^{p-1}\mI(\xi)^{-1}\vg(\vx_i)}. \footnotemark
\end{eqnarray}
\end{thm}
\footnotetext{For D-optimality, $\LEff_{D}(\xi|\xi^*) = \frac{l}{\max\limits_{\vx\in\mathcal{C}}w(\vx)\vg(\vx)^T\mI(\xi)^{-1}\vg(\vx)}.$}

Specifically, for EI-optimality, the lower bound of a design $\xi$ relative to the corresponding EI-optimal design $\xi^*$ is
$$\LEff_{\text{EI}}(\xi|\xi^*) = \frac{\tr(\mI^{-1}(\xi)\mA)}{\max\limits_{\vx\in\mathcal{C}}w(\vx)\vg(\vx)^T\mI(\xi)^{-1}\mA\mI(\xi)^{-1}\vg(\vx)}.$$ Following the stopping criterion used in (\citeO{harman2019randomized}), when the lower bound $\LEff_{\text{EI}}(\xi^r|\xi^*)$ of the efficiency $\Eff_{\text{EI}}(\xi^r|\xi^*)$ reaches a user specified threshold $reqeff$, there is no practical reason to continue the search. We choose $reqeff=0.99$ in all the numerical examples. Another user specified parameter is the maximum iterations allowed, $maxiter$. We choose $maxiter = 100$, and the proposed algorithm converges within 100 iterations in all the numerical examples conducted, while Newton's method fails to converge within 100 iterations for one case in \emph{Example 2}. More details could be found in Section \ref{section:Example}.

Note that the proposed sequential algorithm of constructing EI-optimal design for GLMs does not require the computation of Hessian matrix inversion as in Newton-Raphson methods.
Thus it can avoid issue of singular Hessian matrix and can be computationally efficient than the conventional methods.
The sequential nature of the proposed algorithm also enables efficient search of optimal weights without updating weights for all candidates points. \citeI{martin2015combined} considered a similar sequential algorithm that combines Whittle's method (\citeO{73whi}) with one iteration of multiplicative algorithm to update the weights after a new design point is added.
As a result, the weights are not necessarily optimized in each iteration of the combined algorithm, and the proof of Theorem \ref{thm:cong-algo2} can not be applied to the combined algorithm. 
The convergence property of the combined algorithm for general $\Phi_p$-optimality is not studied in  \citeI{martin2015combined}.

Note that the sequential algorithm (\textbf{Algorithm 2}) can also be easily modified to achieve $\phi_p$-optimal designs, where the analytic formula of directional derivative $\phi(\vx,\xi)$ for $\phi_p$-optimality has been studied in several existing literature (\citeO{06atk}; \citeO{stufken2011optimal}; \citeO{13yang}, etc.).
Moreover, we also establish the convergence property of the proposed sequential algorithm as follows.

\begin{thm}[\textbf{Convergence of Algorithm 2}]\label{thm:cong-algo2}
With a discrete design region $\mathcal{C}$, the design constructed by \textbf{Algorithm  2} converges to EI-optimal design $\xi^*$ that minimizes $EI(\xi,\vbeta,F_{\IMSE})$, as $r\rightarrow\infty$, i.e.,
$$\lim\limits_{r\rightarrow\infty} \EI(\xi^r,\vbeta,F_{\IMSE}) = \EI(\xi^*,\vbeta,F_{\IMSE}).$$
\end{thm}

When the experimental region $\Omega$ is continuous, a discretization would be needed to form the candidate pool $\mathcal{C}$. We would like to point out that the choice of candidate pool $\mathcal{C}$ would affect the efficiency of \textbf{Algorithm 2}, and the computational time increases dramatically as the candidate pool gets large.
We suggest using grid as the candidate pool when the dimension of explanatory variables is small, and choosing the Sobol sequence (\citeO{67sobol}) as the candidate pool when the explanatory variable dimension is large.
The Sobol sequence is a space-filling design that covers the experimental domain $\Omega$ well and is efficient when the dimension of the explanatory variable is high.
To further improve the efficiency of the algorithm, a search strategy inspired by \citeI{12yang} could be employed.
One can start with a more sparse Sobol sequence, and achieve the current best design using Algorithm 2.
Then, we can further create denser and denser candidate pool in the neighborhood of the support points in the current best design until there is no further improvement under the EI-optimality criterion.

\section{Numerical Examples}\label{section:Example}
In this section, we will conduct numerical studies to evaluate the performance of the proposed sequential algorithm (Algorithm 2). As noted in Section \ref{section: Algorithm 2}, Algorithm 2 can be applied to construct $\Phi_p$-optimal designs such as A-optimal design.
Based on our best knowledge, there is no existing algorithm that can directly construct the EI-optimal design for GLMs. However, the Newton-type method (\citeO{13yang}) could be revised to construct EI-optimal design. The proposed algorithm (Algorithm 2) will be compared with the Newton-type method in \citeI{13yang}, which adopts Newton's method to update the weights of design points and is an efficient algorithm in the literature. The comparison will be conducted under different generalized linear models with various settings of variable dimensions. Both algorithms are implemented in \textsc{Matlab}, and the code for Newton's method is converted from the SAS code kindly provided by the authors of \citeI{13yang}. All codes were run on a MacBook Pro with 2.4 GHz Intel Core i5 processor. The Newton-type method requires a well-conditioned Hessian matrix to update the weights, but the Hessian matrix could be numerically singular in a certain iteration and results in inaccurate weight updates. In this numerical example, the generalized inverse of Hessian matrix is used for the Newton-type method. The proposed algorithm always returns nonnegative weights, only eliminates the design points with almost zero weight and does not require Hessian matrix inversion. Grid candidate pool of size $N = (s+1)^d$ with $s+1$ equally spaced points for each explanatory variable $x_i$, $i=1,\cdots,d$ is used in all numerical examples.
Note that our \textbf{Algorithm 2} does not require the candidate pool to be a grid. We would suggest using space-filling design such as Sobol sequence as the candidate pool for high-dimensional explanatory variable.
The target efficiency lower bound $reqeff$ in the stopping criterion is chosen to be 0.99, and we consider two designs are both good enough if their efficiency lower bounds exceed 0.99.
Thus, it is fair to compare the efficiency of the algorithms based on the computational time.

\emph{Example 1}.  The setting of this example  follows the Example 2 in (\citeO{13yang}), which
considers the linear model
\begin{eqnarray*}
Y &\sim & \theta_1+\theta_2x_1+\theta_3x_1^2+\theta_4x_2+\theta_5x_1x_2+N(0,\sigma^2),\\
\boldsymbol{\theta} &=& (\theta_1,\theta_2,\theta_3,\theta_4,\theta_5),
\end{eqnarray*}
where $\Omega = \{(2i/s-1, j/s)\}$, $i=0,1,...,s, j = 0,1,...,s\}$, where $s$ is the number of grid points in each variable and the total number of points in $\Omega$ is $N=(s+1)^2$. Note that since the experimental region is discrete, it is not necessary to discretize the experimental region.
The proposed algorithm and the Newton-type method are compared under A- and EI-optimality. For EI-optimality, we choose $F_{\IMSE} = F_{\unif}$.

Here we consider the same initial design, consisting of five randomly chosen points in $\Omega$, for both Newton-type method and the proposed algorithm.
For different values of $s$, Figure \ref{fig:LMA} reports the average computational time (in seconds) of the two algorithms and average efficiency $\Eff(\xi^*_{\text{Proposed}}|\xi^*_{\text{Newton}})$ based on 10 randomly chosen initial designs, where $\xi^*_{\text{Proposed}}$ and $\xi^*_{\text{Newton}}$ are the optimal designs achieved using the proposed algorithm and Newton's method, respectively.
\begin{figure}[htpb]
\centering
\subfloat[EI-Optimality, Candidate Pool Size $N = (s+1)^2$]
{{\includegraphics[width=6cm]{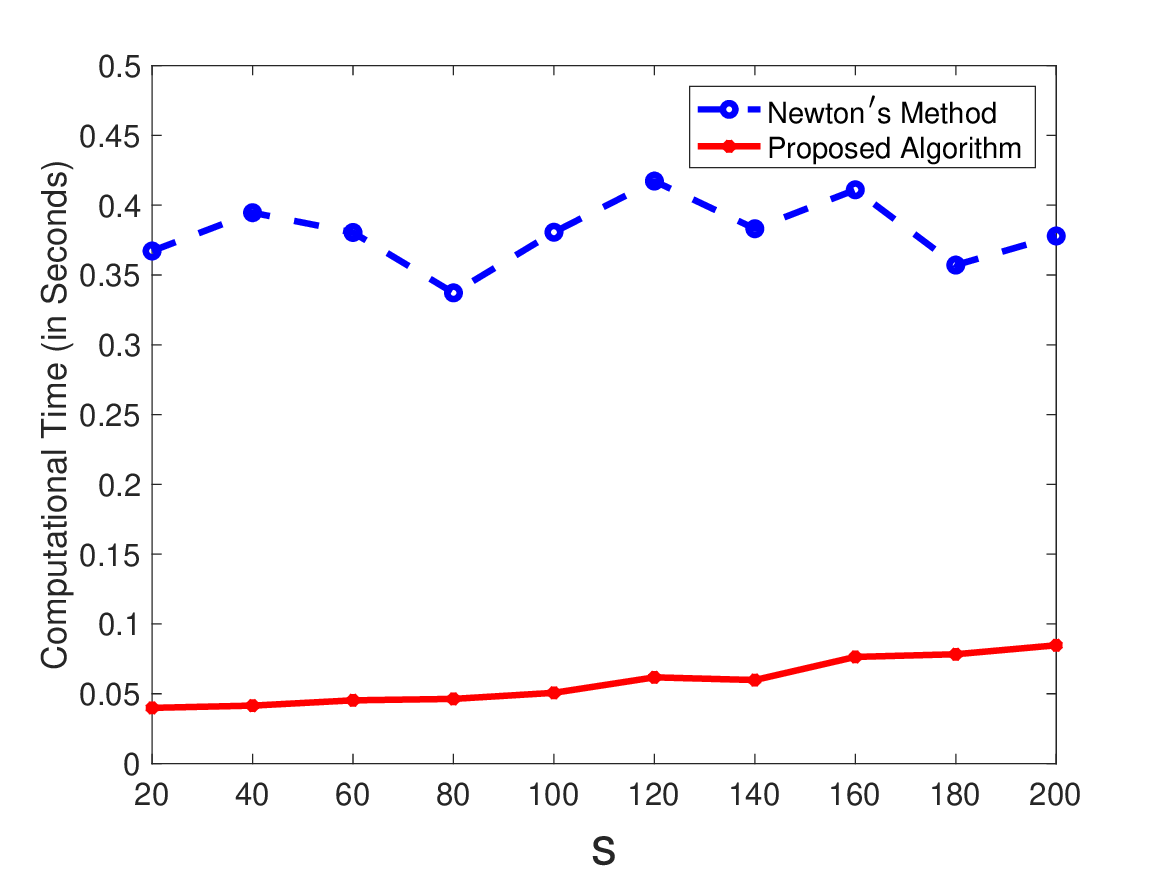}}
\hfill
{\includegraphics[width=6cm]{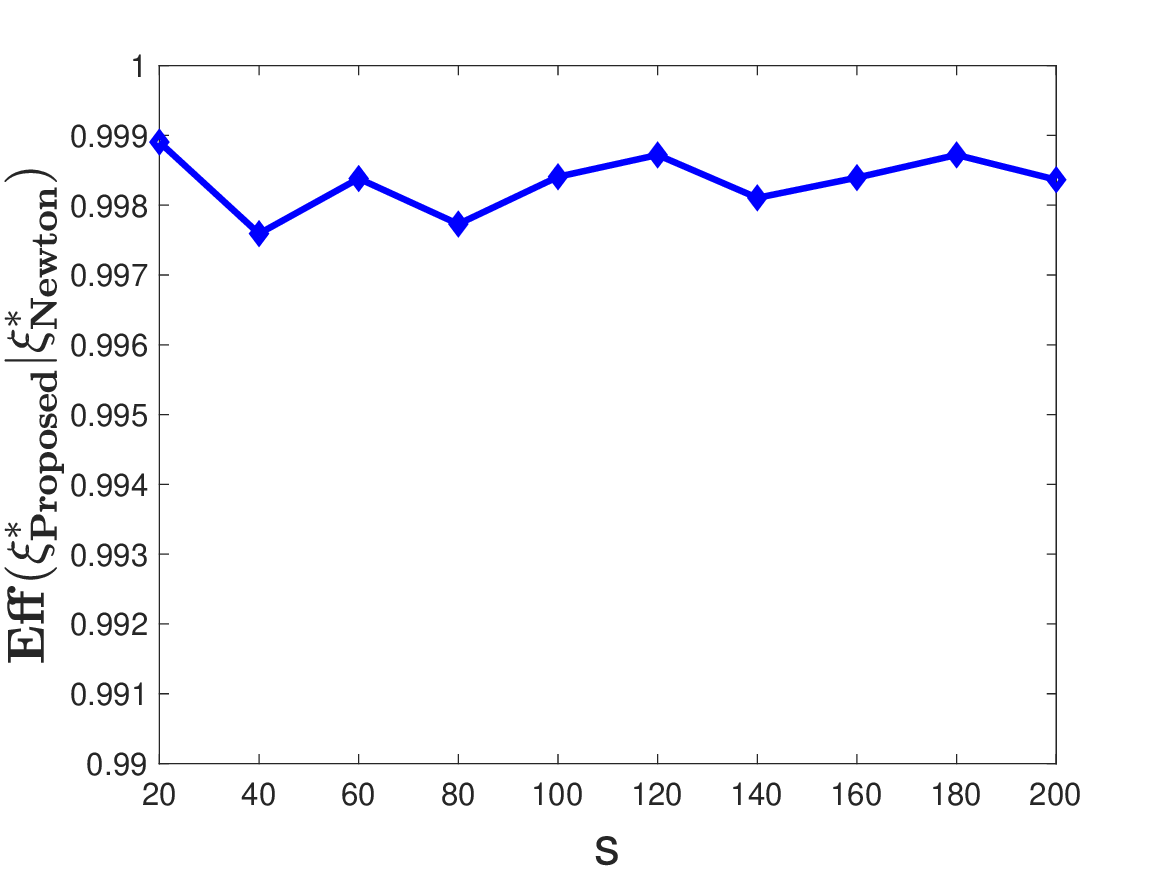}}}
\qquad
\subfloat[A-Optimality, Candidate Pool Size $N = (s+1)^2$]
{{\includegraphics[width=6cm]{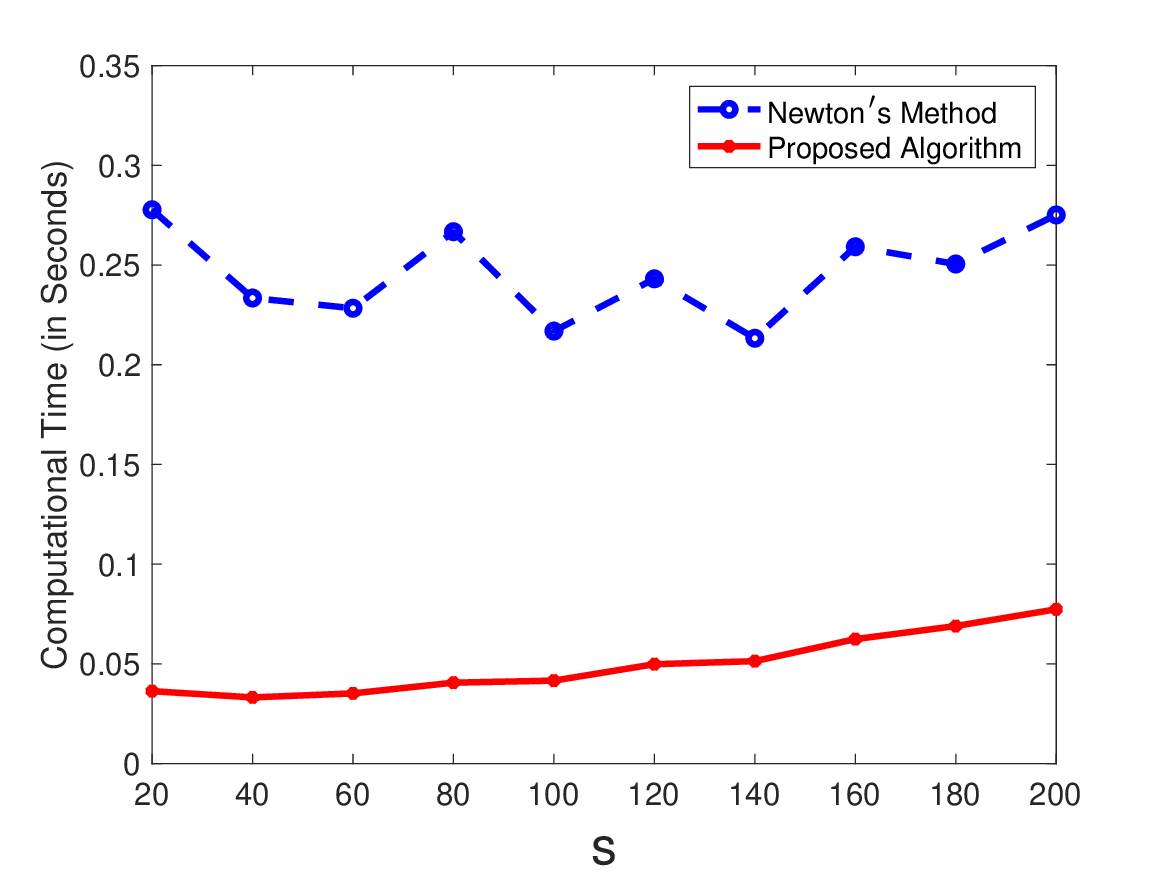}}
\hfill
{\includegraphics[width=6cm]{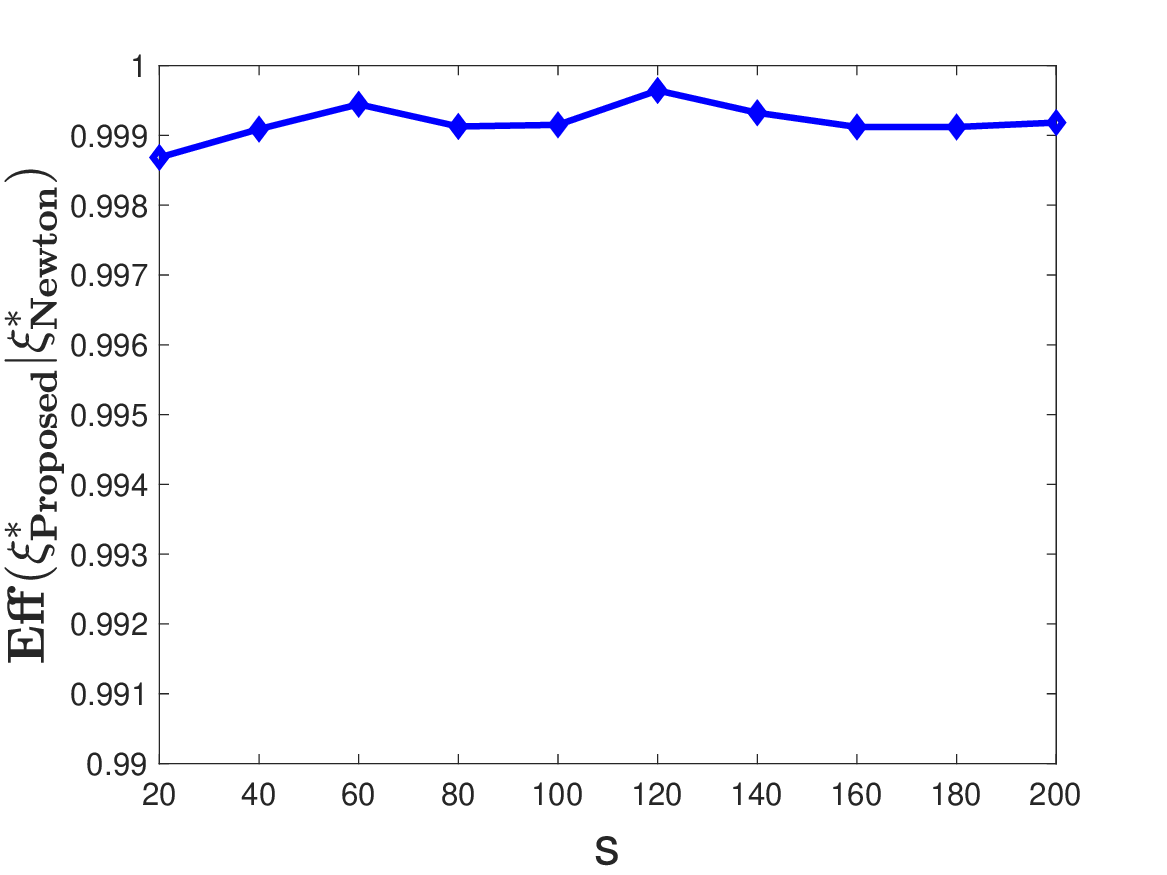}}}
\caption{Computational Time and Relative Efficiency $\Eff(\xi^*_{\text{Proposed}}|\xi^*_{\text{Newton}})$ for A- and EI-optimal Designs with Linear Regression Model}
\label{fig:LMA}
\end{figure}

Figure \ref{fig:LMA} shows that the proposed algorithm is more efficient than the Newton-type method since it does not requires additional remedy for negative weights, and the multiplicative algorithm is only performed on a small set of design points. It is found that the computational time of Newton's method largely varies and significantly depends on the initial design points. In contrast,
the proposed algorithm is very robust to initial design points and has relatively stable computational performance. Although the efficiencies of both designs exceed 0.99, the design achieved by Newton's method has a slightly smaller optimality criterion value, and this could be due to the quadratic convergence rate of Newton's method. As shown in Figure \ref{fig:LMA}, the efficiency of the design achieved by the proposed algorithm relative to that achieved by Newton's method is always above 99.7\%, but the proposed algorithm is about 4 times faster than the Newton's method.

We further compare the above EI-optimal design with $F_{\IMSE}(\vx) = F_{\unif}(\vx)$ to the EI-optimal design with a different $F_{\IMSE}(\vx) = F^{(1)}_{\asine}(x_1)F^{(2)}_{\asine}(x_2)$, where $F_{\asine}^{(1)}(x_1)$ and $F_{\asine}^{(2)}(x_2)$ are arcsine distributions on $[-1,1]$ and $[0,1]$, respectively. Different from the uniform distribution, arcsine distribution on bounded support $[a,b]$ has probability density function $\rho(x) = \frac{1}{\pi\sqrt{(x-a)(x-b)}}$, $x\in[a,b]$, which puts more weight towards the interval ends. Figure \ref{fig:unifVsarcsine} shows the support points of EI-optimal designs when $F_{\IMSE}$ is chosen to be uniform distribution and arcsine distribution.
\begin{figure}[htpb]
\centering
{\includegraphics[width=7cm]{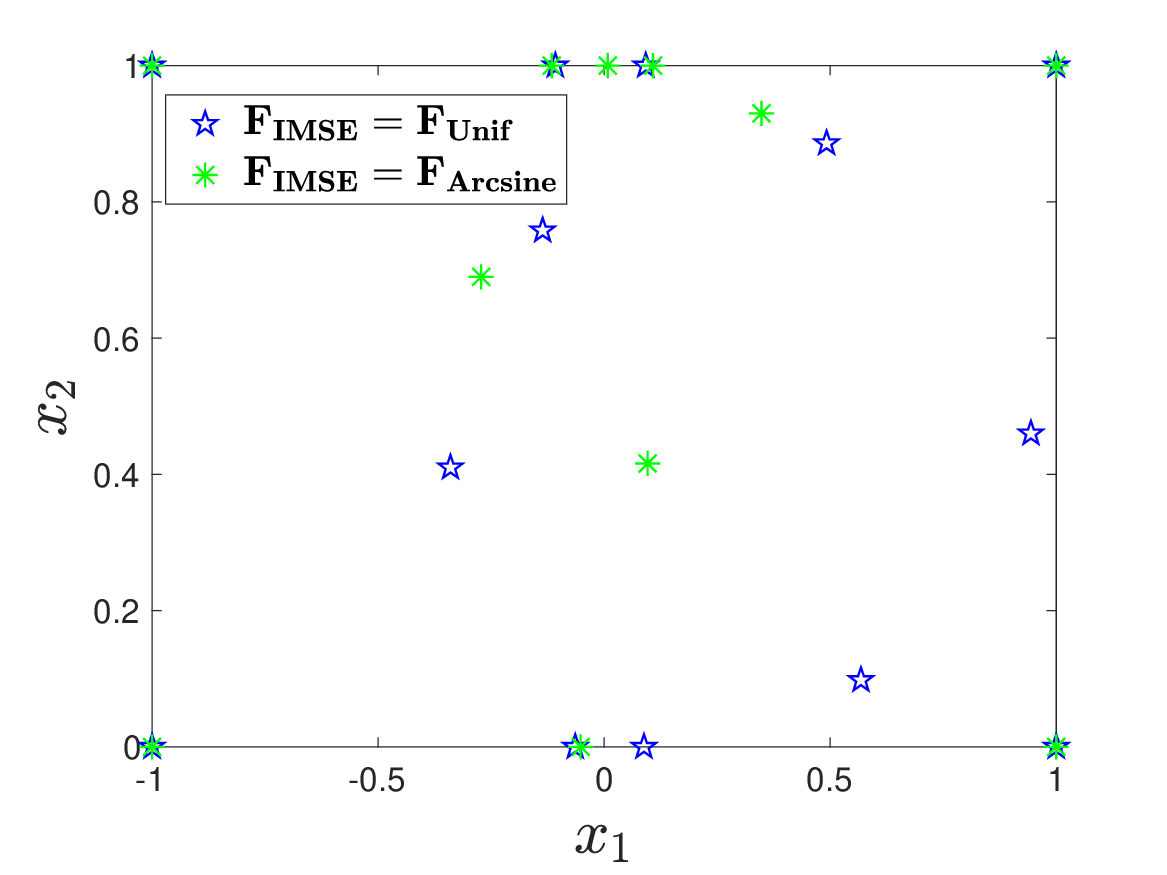}}
\caption{Support Points of EI-Optimal Designs with $F_{\IMSE} = F_{\unif}$ and $F_{\IMSE} = F_{\asine}$}
\label{fig:unifVsarcsine}
\end{figure}
Under $F_{\IMSE} = F_{\unif}$, The efficiency of EI-optimal design achieved using arcsine distribution relative to that achieved using uniform distribution is 0.9564. On the other hand, under $F_{\IMSE} = F_{\asine}$, The  efficiency of EI-optimal design achieved using uniform distribution relative to that achieved using arcsine distribution is 0.9595.

\emph{Example 2.} This example considers the logistic regression model with binary response.
Assume that the domain of $d-$dimensional explanatory variable $\vx = [x_1,...,x_d]$ is standardized to be a unit hypercube $[-1,1]^d$.
With $l$ basis functions $\vg(\vx) = [g_1(\vx),...g_l(\vx)]^T$ and regression coefficients $\vbeta=[\beta_1,...,\beta_l]^T$, the logistic regression model with binary response $Y \in \{0, 1\}$ is defined as:
$$Prob(Y=1|\vx) = \frac{e^{\vbeta^T\vg(\vx)}}{1+e^{\vbeta^T\vg(\vx)}}.$$
Usually the basis functions $\vg(\vx)$ are low degree polynomials of explanatory variable $\vx$, and in this example we consider linear predictors, i.e., $\vg(\vx) = (1, \vx^T)^T = (1, x_1,...,x_d)^T$.
Given a design $\xi = \left\{\begin{array}{ccc}
\vx_1,&...,&\vx_n\\
\lambda_1,&...,&\lambda_n
\end{array}\right \}$, the EI-optimality criterion with some probability distribution $F_{\IMSE}$ is
$$\EI(\xi,\vbeta,F_{\IMSE}) = \tr\left(\mA \mI(\xi)^{-1}\right),$$
where $\mA = \int_{\Omega}(1, \vx^T)^T(1, \vx^T) v(\vx)\dif F_{\IMSE}(\vx)$ with $v(\vx) = e^{2(\beta_{0}+\vbeta^{T} \vx)} / (1 + e^{\beta_{0}+\vbeta^{T} \vx})^{4}$,
and
$\mI(\xi) = \sum\limits_{i=1}^n \lambda_i w(\vx_{i})(1, \vx_{i}^T)^T(1, \vx_{i}^T)$ with
$w(\vx_{i}) = e^{\beta_{0}+\vbeta^{T} \vx_{i}} / (1 + e^{\beta_{0}+\vbeta^{T} \vx_{i}})^{2}$.
In this example, we consider ({\it case i}) the classical I-optimality with $F_{\IMSE}$ being a uniform distribution on $\Omega_c = \Omega = [-1,1]^d$, i.e., $F'_{\IMSE}(\vx) = \frac{1}{2^d}, \vx\in\Omega$ and $\mA = \int_{[-1,1]^d}\frac{1}{2^d}(1, \vx^T)^T(1, \vx^T)v(\vx)\dif\vx$;
and ({\it case ii}) $F_{\IMSE}$ being a uniform distribution on positive half hypercube $\Omega_c = [0,1]^d\subset\Omega = [-1,1]^d$, i.e., $F'_{\IMSE}(\vx) = \left\{\begin{array}{ll} 1, & \vx\in\Omega_c,\\ 0, & \vx\notin\Omega_c\end{array}\right.$ and $\mA = \int_{[0,1]^d}(1, \vx^T)^T(1, \vx^T)v(\vx)\dif\vx$. To investigate the performance of the proposed algorithm when the explanatory variable dimension $d$ gets large, we compute the EI-optimal design under the following scenarios:
\begin{enumerate}[(a)~]
\item
$d = 1$, $\vbeta = [0.2, 1.6]^T$
\item
$d=2$, $\vbeta = [2, 1, -2.5]^T$.
\item
$d=3$, $\vbeta = [0.5, 1.6, -2.5, 2]^T$.
\end{enumerate}
We explore three properties of the proposed algorithm: (1) efficiency compared to Newton's method; (2) choice of convergence rate $\delta$ in Algorithm 1; and (3) choice of candidate pool size.
\begin{itemize}
\item
\emph{Efficiency compared to Newton's method.}

The computational time comparison between the proposed algorithm and Newton's method is summarized in Figures \ref{fig:multiI-eg2a} and \ref{fig:multiI-eg2b}.
\begin{figure}[htpb]
\centering
\subfloat[$d=1$, case i, Candidate Pool Size $N = s+1$]
{{\includegraphics[width=6cm]{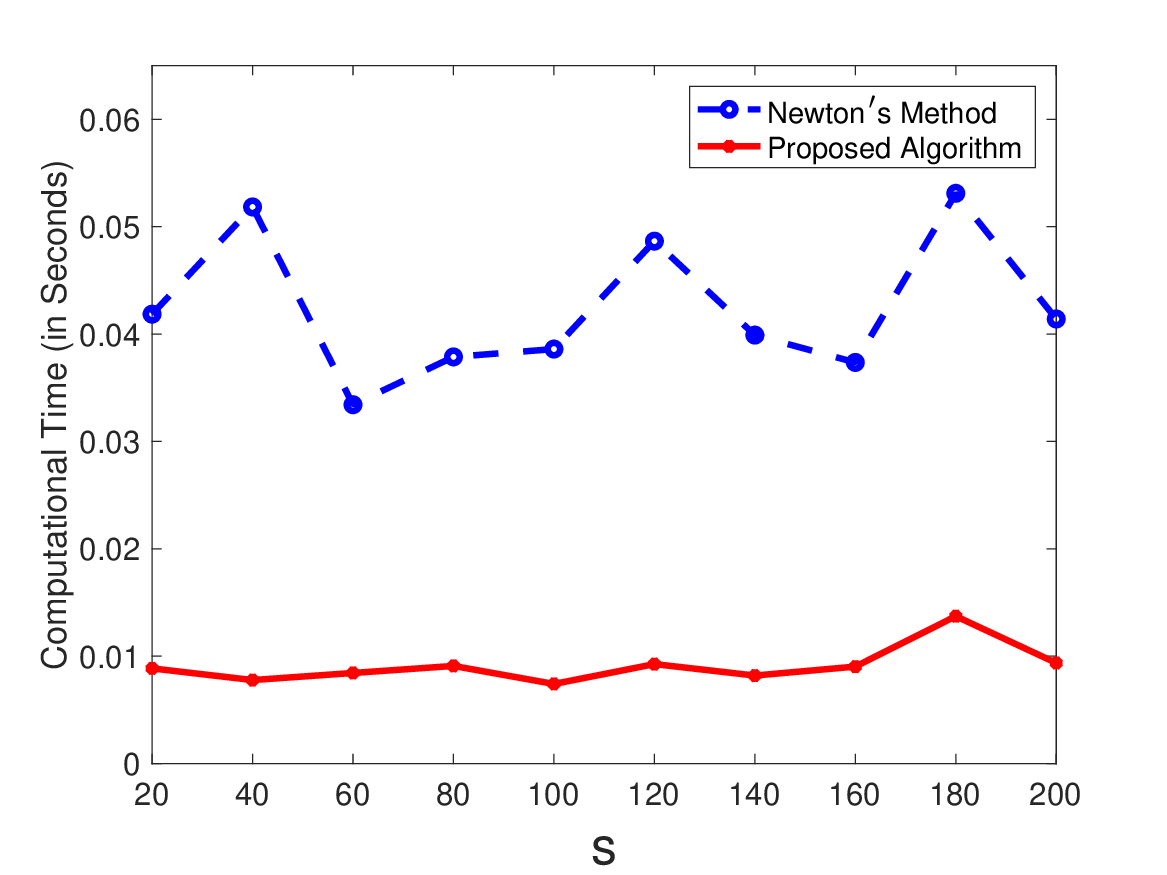}}
\hfill
{\includegraphics[width=6cm]{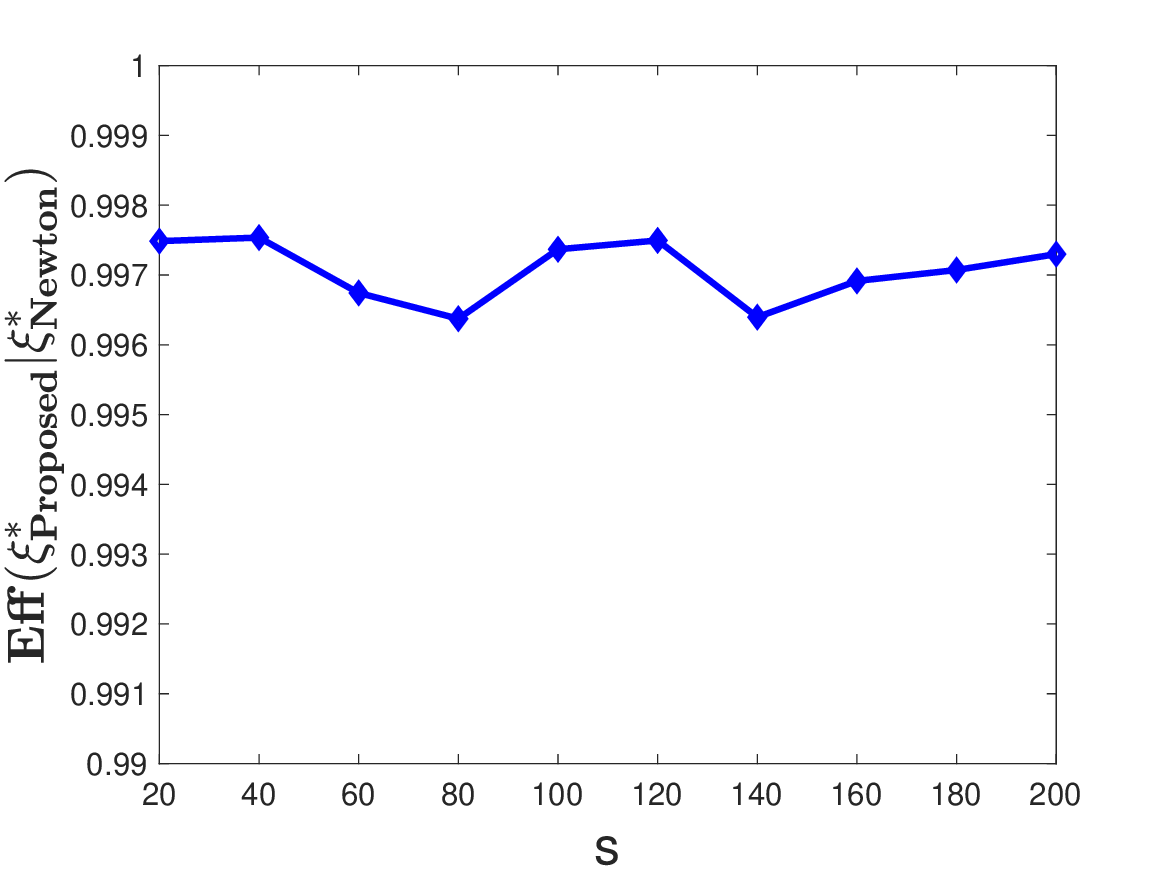}}}
\qquad
\subfloat[$d=1$, case ii, Candidate Pool Size $N = s+1$]
{{\includegraphics[width=6cm]{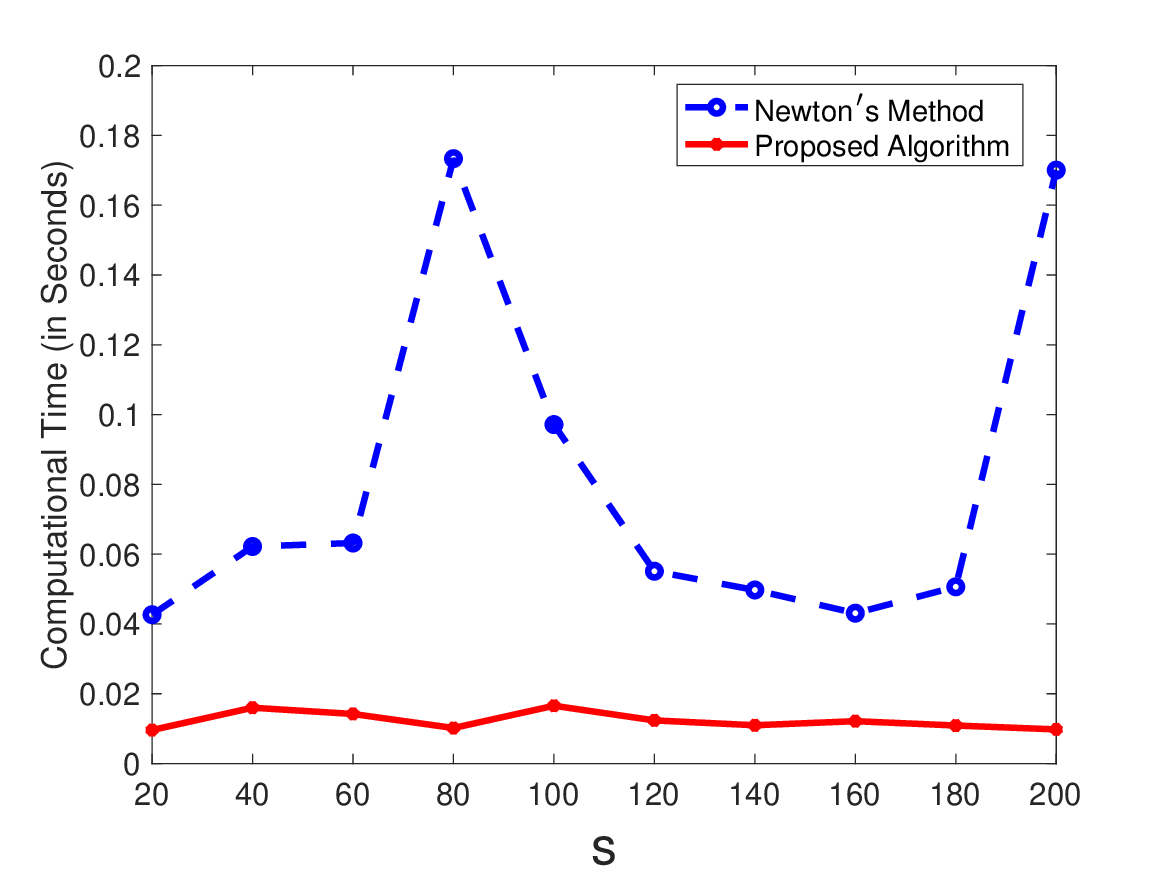}}
\hfill
{\includegraphics[width=6cm]{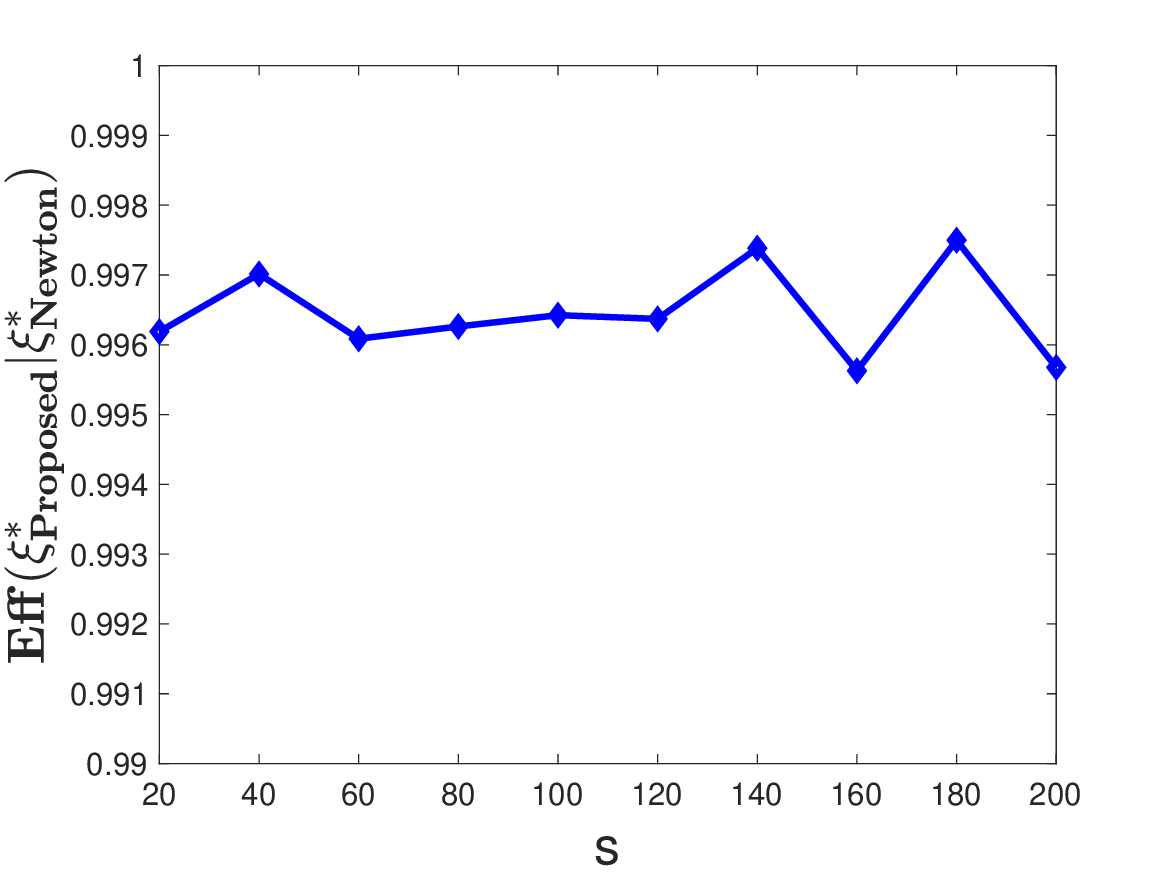}}}
\qquad
\subfloat[$d=2$, case i, Candidate Pool Size $N = (s+1)^2$]
{{\includegraphics[width=6cm]{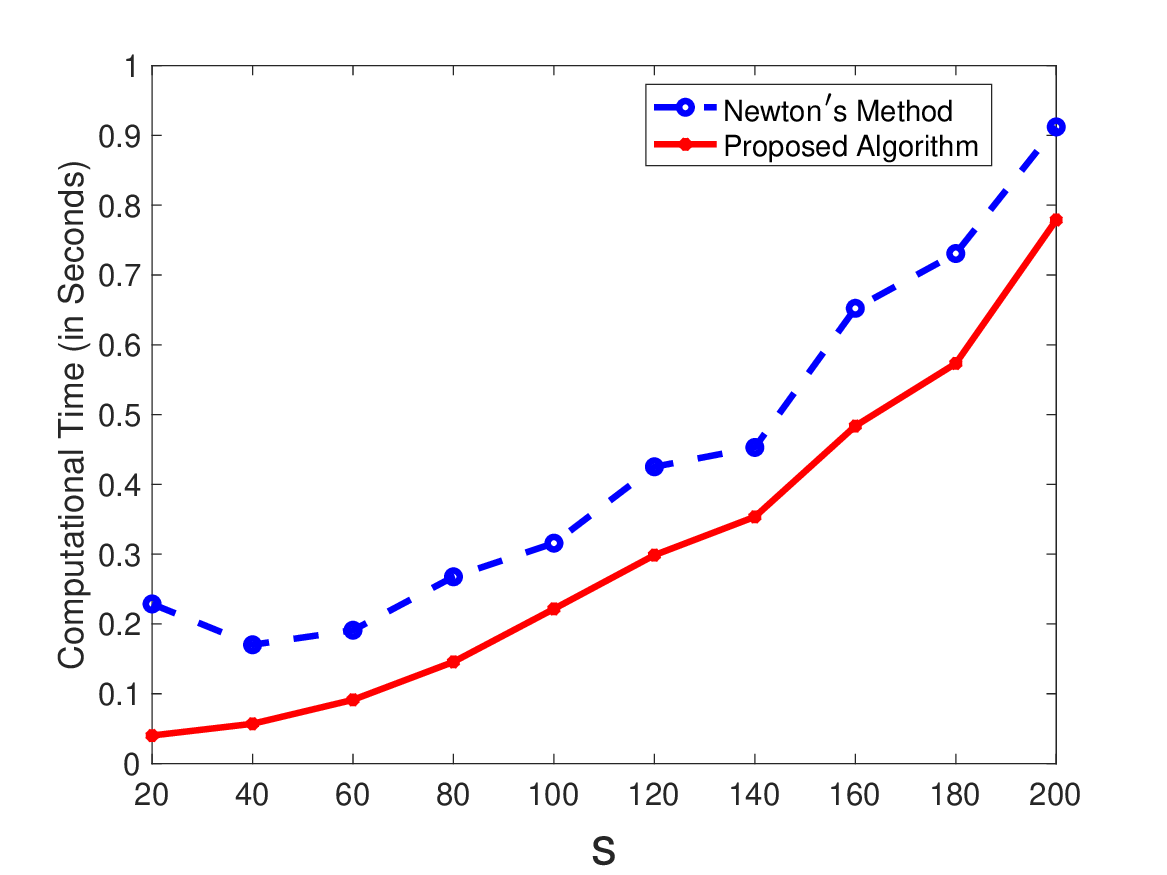}}
\hfill
{\includegraphics[width=6cm]{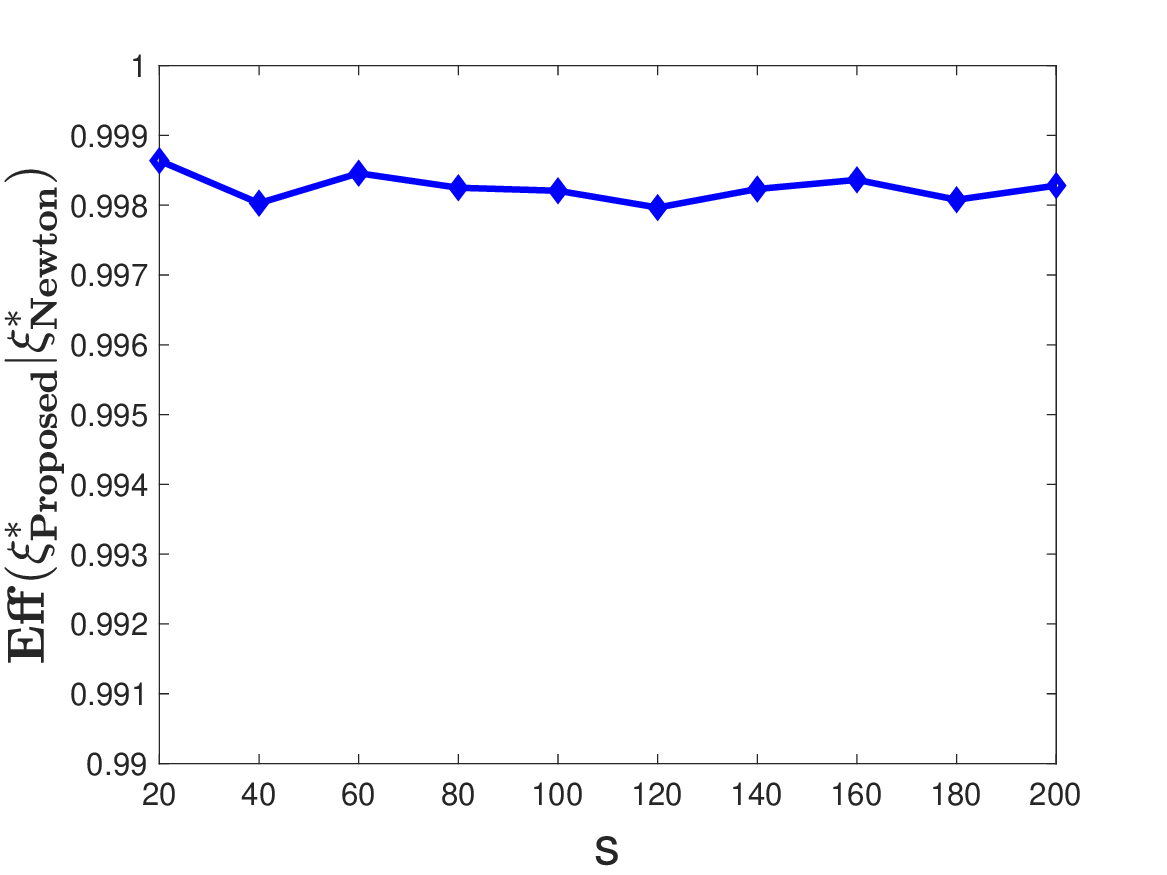}}}
\caption{Computational Time (in Seconds) and Relative Efficiency $\Eff(\xi^*_{\text{Proposed}}|\xi^*_{\text{Newton}})$ for EI-Optimal Designs with $d$-Dimensional Logistic Regression Model}
\label{fig:multiI-eg2a}
\end{figure}
\begin{figure}[htpb]
\centering
\subfloat[$d=2$, case ii, Candidate Pool Size $N = (s+1)^2$]
{{\includegraphics[width=6cm]{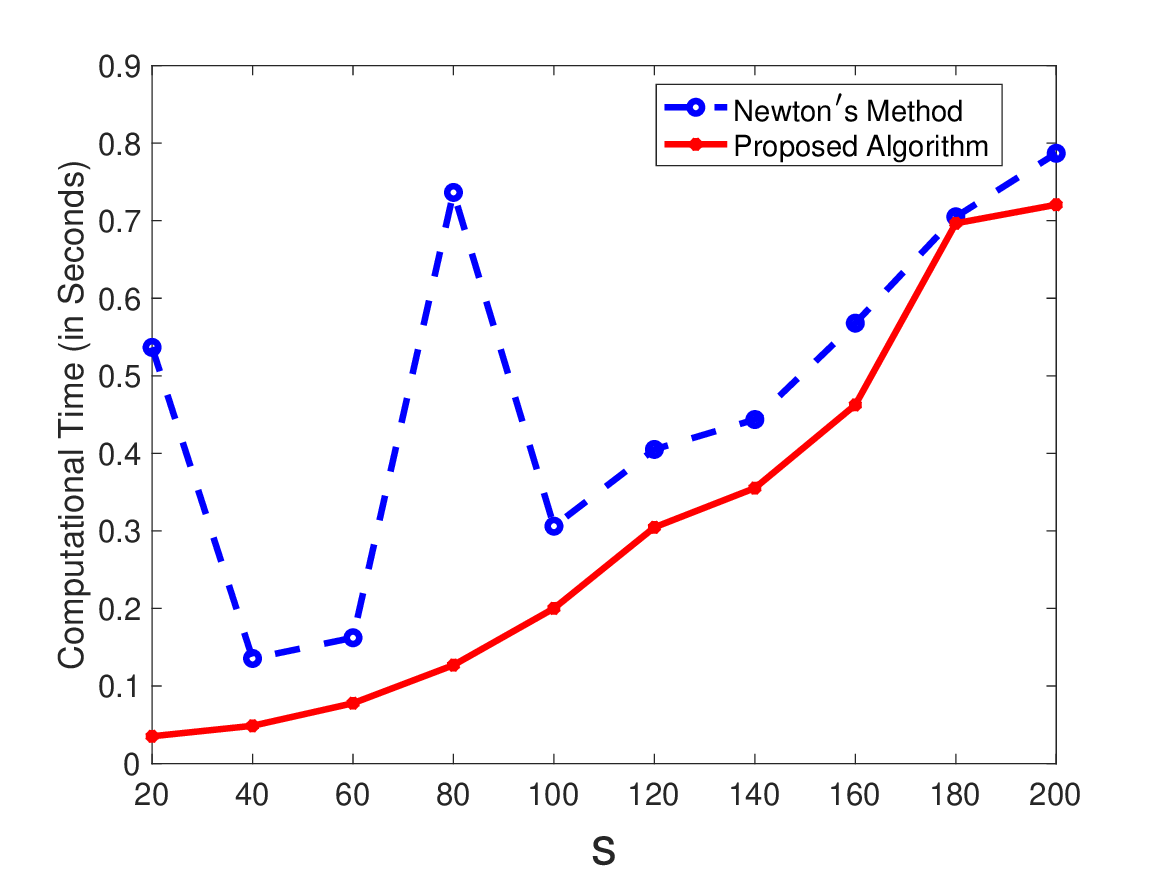}}
\hfill
{\includegraphics[width=6cm]{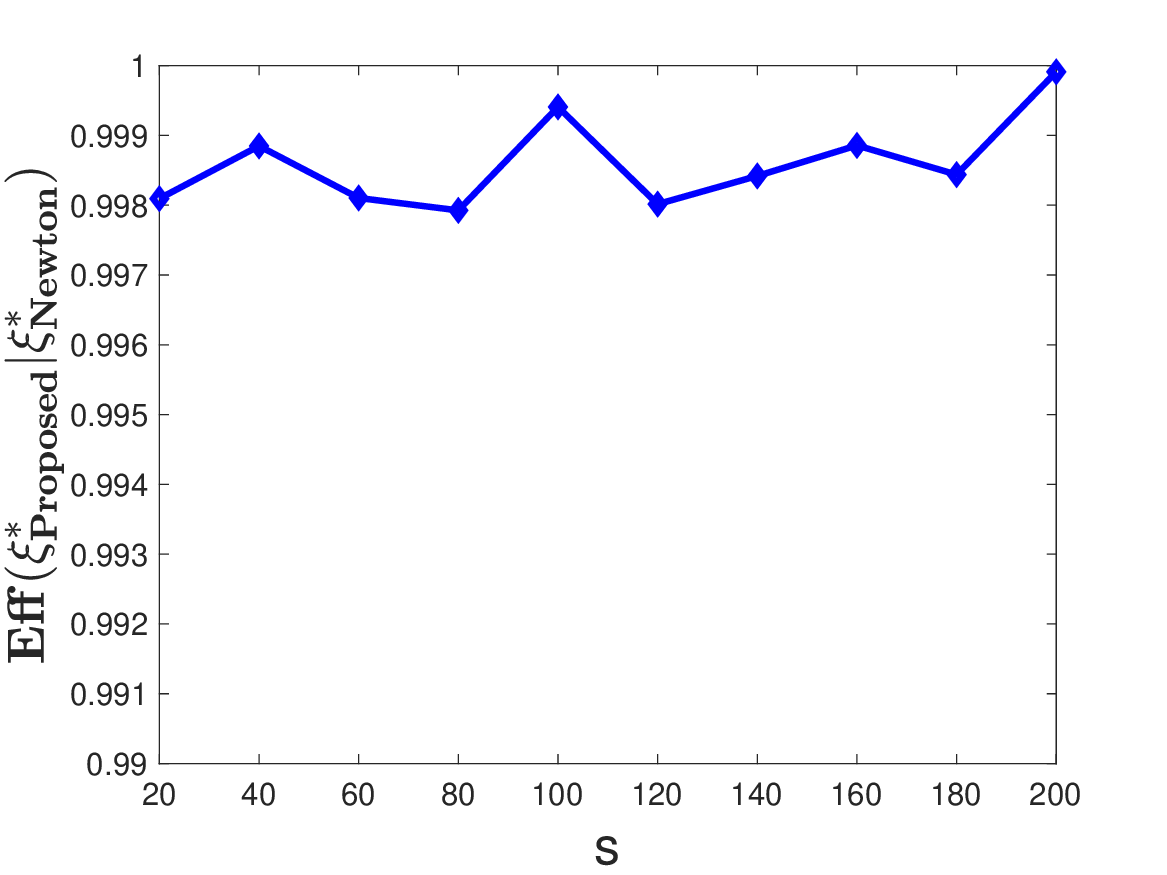}}}
\qquad
\subfloat[$d=3$, case i, Candidate Pool Size $N = (s+1)^3$]
{{\includegraphics[width=6cm]{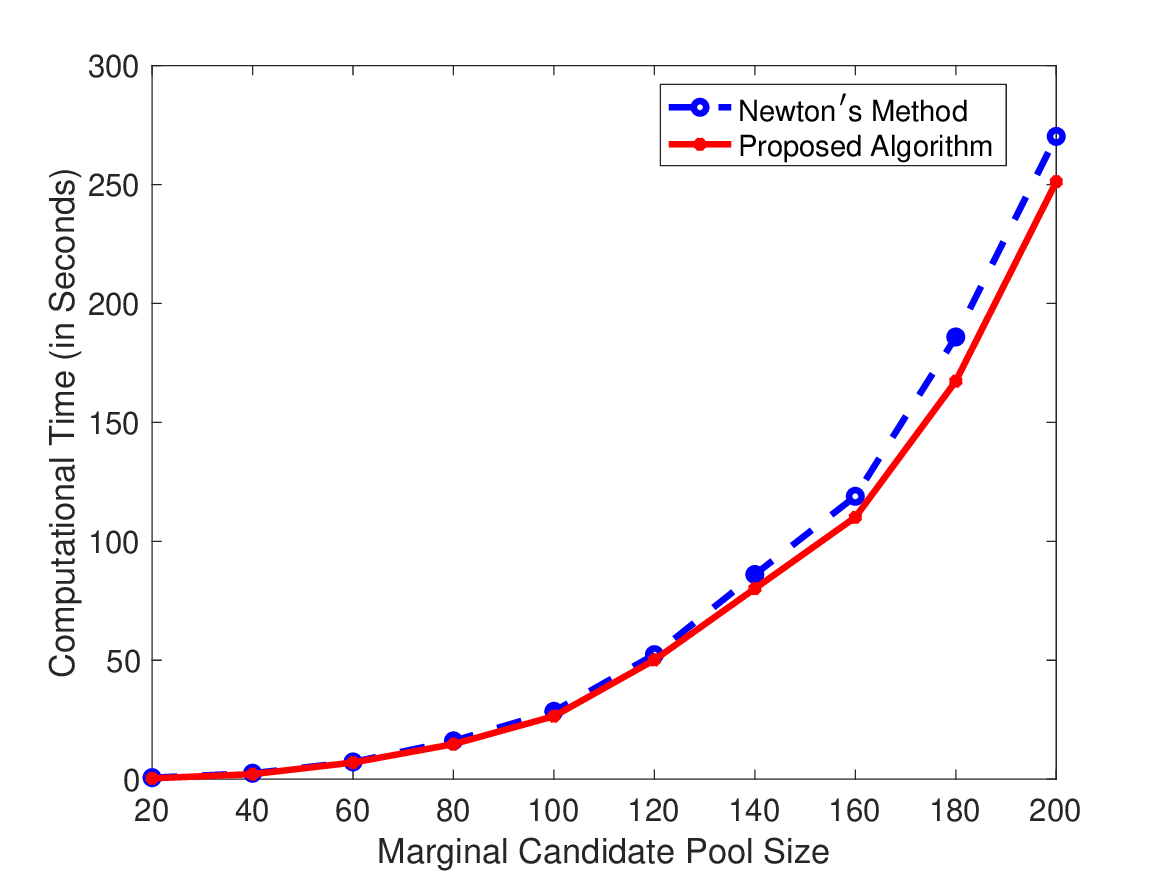}}
\hfill
{\includegraphics[width=6cm]{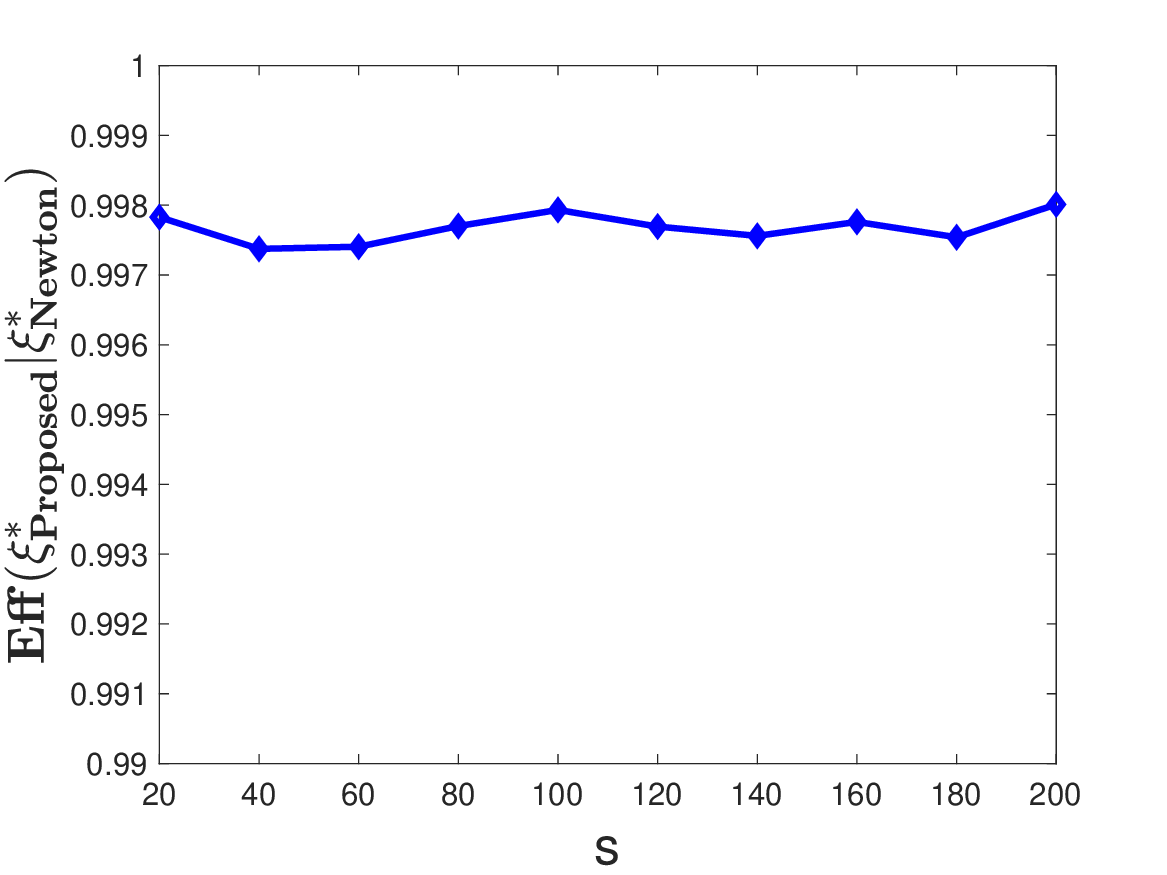}}}
\caption{Computational Time (in Seconds) and Relative Efficiency $\Eff(\xi^*_{\text{Proposed}}|\xi^*_{\text{Newton}})$ for EI-Optimal Designs with $d$-Dimensional Logistic Regression Model (Continued)}
\label{fig:multiI-eg2b}
\end{figure}
For $d=3$ of case ii,
there are situations that Newton's method does not converge in 100 iterations, i.e., the efficiency lower bound does not reach 0.99 in 100 iterations.
In contrast, the proposed method is much more stable and meets the efficiency criterion in all scenarios.
Thus, for $d=3$, we only report the computational time comparison for case i where both algorithms converge within 100 iterations.
Clearly, the proposed method is quite computationally efficient.
When the candidate pool size gets larger, Figure \ref{fig:multiI-eg2b}(c) shows that the computational times of two algorithms become close with the proposed algorithm to be slightly faster.
Note that both algorithms are Wynn-Fedorov type algorithms, and the computational time is dominated by evaluating directional derivative over the large candidate pool.
As the candidate pool gets large, the computational advantage of the proposed algorithm would become less pronounced.
However, it is worth remarking that even with similar computational times,
the proposed algorithm still has other advantages such as guaranteed convergence and simple implementation.

\item
\emph{Choice of convergence rate $\delta$.}

As discussed in Section \ref{subsec:mulalg}, the convergence parameter $\delta\in(0,1)$ in \textbf{Algorithm 1} is user defined and we choose $\delta=0.5$ for EI-optimality. We further explore the computational time of the proposed algorithm with different choices of $\delta$ ranging from 0.05 to 0.95 in $d=2$ of case i. The results are shown in Figure \ref{fig:multiI-eg2d}. It can be seen that the computational time of the proposed algorithm is quite robust to the choice $\delta$ as long as $\delta$ is not too small.
\begin{figure}[htpb]
\centering
{\includegraphics[width=7cm]{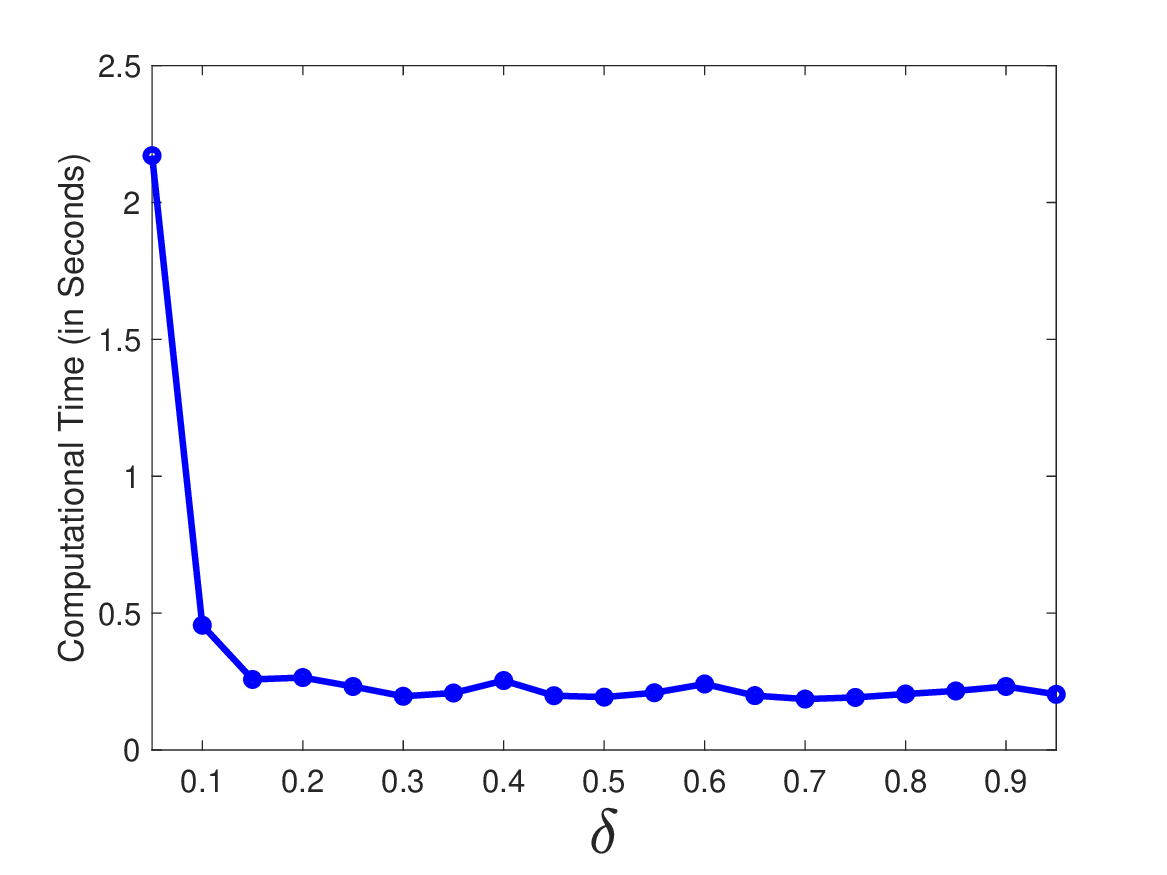}}
\caption{Computational Time (in Seconds) of the Proposed Algorithm with Different Convergence Parameter $\delta$}
\label{fig:multiI-eg2d}
\end{figure}

\item
\emph{Choice of candidate pool size.}

Since a discrete candidate pool $\mathcal{C}$ is required in \textbf{Algorithm 2}, discretization is needed when the experimental region is continuous.
When using grid as the candidate pool, its size can increase exponentially as the dimension $d$ and the number of grid points in each dimension increases,
and correspondingly the computational time of the proposed algorithm also increases dramatically.
We explore how the size of candidate pool would influence the efficiency of the achieved design via $d=2$ case i.
The EI-optimal designs are constructed when the candidate pool size varies from $N=10^2$ to $N=1180^2$.
The efficiencies of the optimal designs relative to the optimal design constructed using $N=1180^2$ are reported in Figure 5.
It is not surprising that the relative efficiency does not increase monotonically as the candidate pool size increases.
One explanation is that the quality of the obtained design depends on whether the candidate pool contains the support points of the optimal design with continuous experimental region, but not the size of the candidate pool.
Although the candidate pool size ranges widely from $N=10^2$ to $N=1180^2$, the relative efficiencies are all very close to 1 (higher than 0.999), which indicates that a very fine grid candidate pool may not be necessary to construct an EI-optimal design.
Based on this observation, we would suggest using Sobol sequence as candidate design points for our proposed method when the dimension of design region $d$ is large.
\begin{figure}[htpb]
\centering
{\includegraphics[width=7cm]{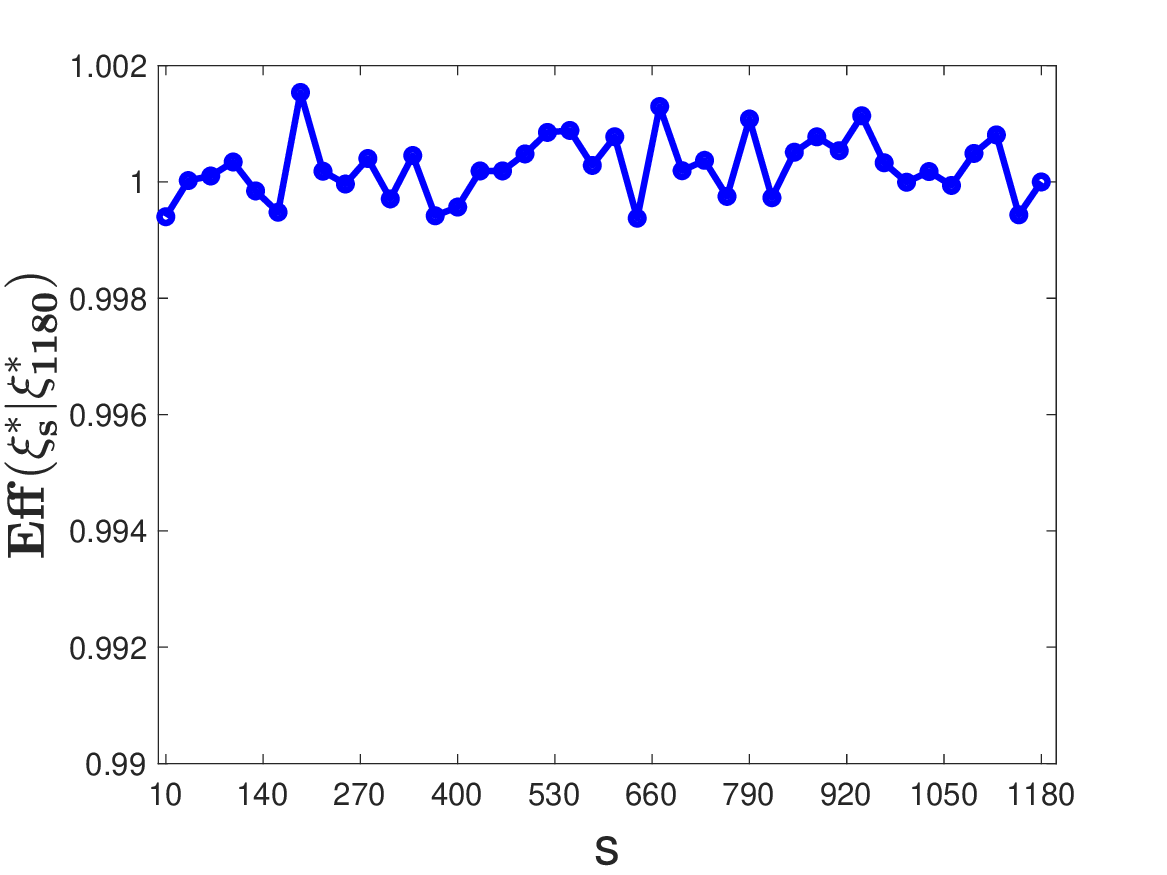}}
\caption{Relative Efficiency of EI-Optimal Designs with Different Candidate Pool Sizes $N = (s+1)^2$}
\label{fig:multiI-eg2c}
\end{figure}
\end{itemize}

%

\emph{Example 3}. This example considers the Poisson regression models, which is a popular statistical tool to model count data in many applications (e.g., \citeO{gart1964analysis}; \citeO{el1987application}).
Like in \emph{Example 2}, we assume that the domain of $d-$dimensional explanatory variable $\vx = [x_1,...,x_d]$ is standardized to be a unit hypercube $\Omega = [-1,1]^d$.
With $l$ basis functions $\vg(\vx) = [g_1(\vx),...g_l(\vx)]^T$ and regression coefficients $\vbeta=[\beta_1,...,\beta_l]^T$, the Poisson regression model with count response $Y \in \{0, 1,...,\}$ has the mean function:
$$\mu(\vx) = \E[Y(\vx)] = e^{\vbeta^T\vg(\vx)}.$$
Similar as \emph{Example 2}, we consider linear predictors, i.e., $\vg(\vx) = (1, \vx^T)^T = (1, x_1,...,x_d)^T$.
Given a design $\xi = \left \{\begin{array}{ccc}
\vx_1,&...,&\vx_n\\
\lambda_1,&...,&\lambda_n
\end{array}\right \}$, the classical I-optimality criterion with $F_{\IMSE} = F_{\unif}$ on $[-1,1]^d$ is
$$\EI(\xi,\vbeta,F_{\unif}) = \tr\left(\mA \mI(\xi)^{-1}\right),$$
with $\mA = \int\frac{1}{2^d}(1, \vx^T)^T(1, \vx^T)e^{2\left(\beta_0+ \vbeta^{T} \vx\right)}\dif \vx$ and $\mI(\xi) = \sum\limits_{i=1}^n \lambda_i e^{\beta_0+ \vbeta^{T} \vx_{i}}(1, \vx_{i}^T)^T(1, \vx_{i}^T)$.
We consider the same scenarios as in \emph{Example 2} for Poisson regression model and compute the corresponding classical I-optimal designs:
\begin{enumerate}[(a)~]
\item
$d = 1$, $\vbeta = [0.2, 1.6]^T$
\item
$d=2$, $\vbeta = [2, 1, -2.5]^T$.
\item
$d=3$, $\vbeta = [0.5, 1.6, -2.5, 2]^T$.
\end{enumerate}
Here we also use the grids to form the set of candidate design points. The computation comparison is shown in Figure \ref{fig:multiI-eg3}.
Similar to the results in \emph{Example 2}, when the candidate pool size is moderate, the proposed algorithm outperforms the Newton's method by a large margin regarding the computational efficiency, and also preserves a high design efficiency.

\begin{figure}[htpb]
\centering
\subfloat[$d=1$,  Candidate Pool Size $N = s+1$]
{{\includegraphics[width=6cm]{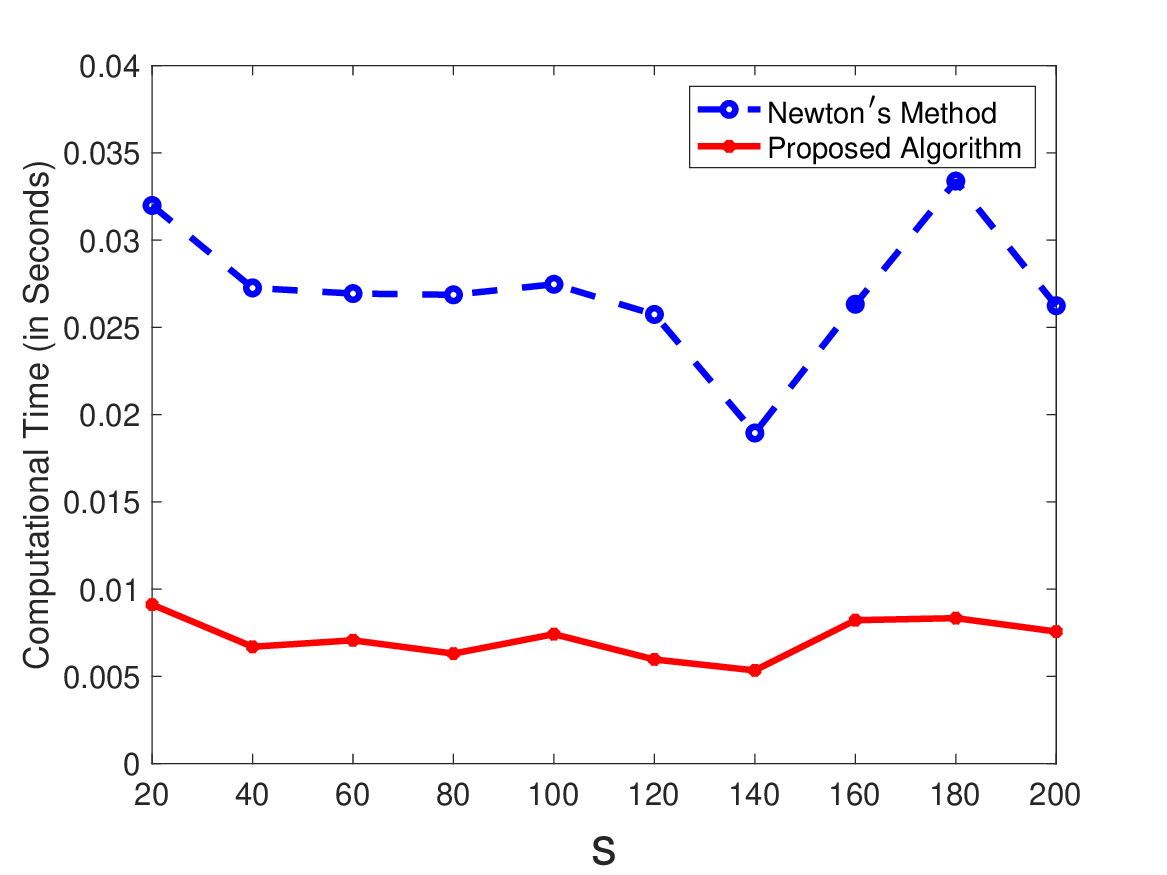}}
\hfill
{\includegraphics[width=6cm]{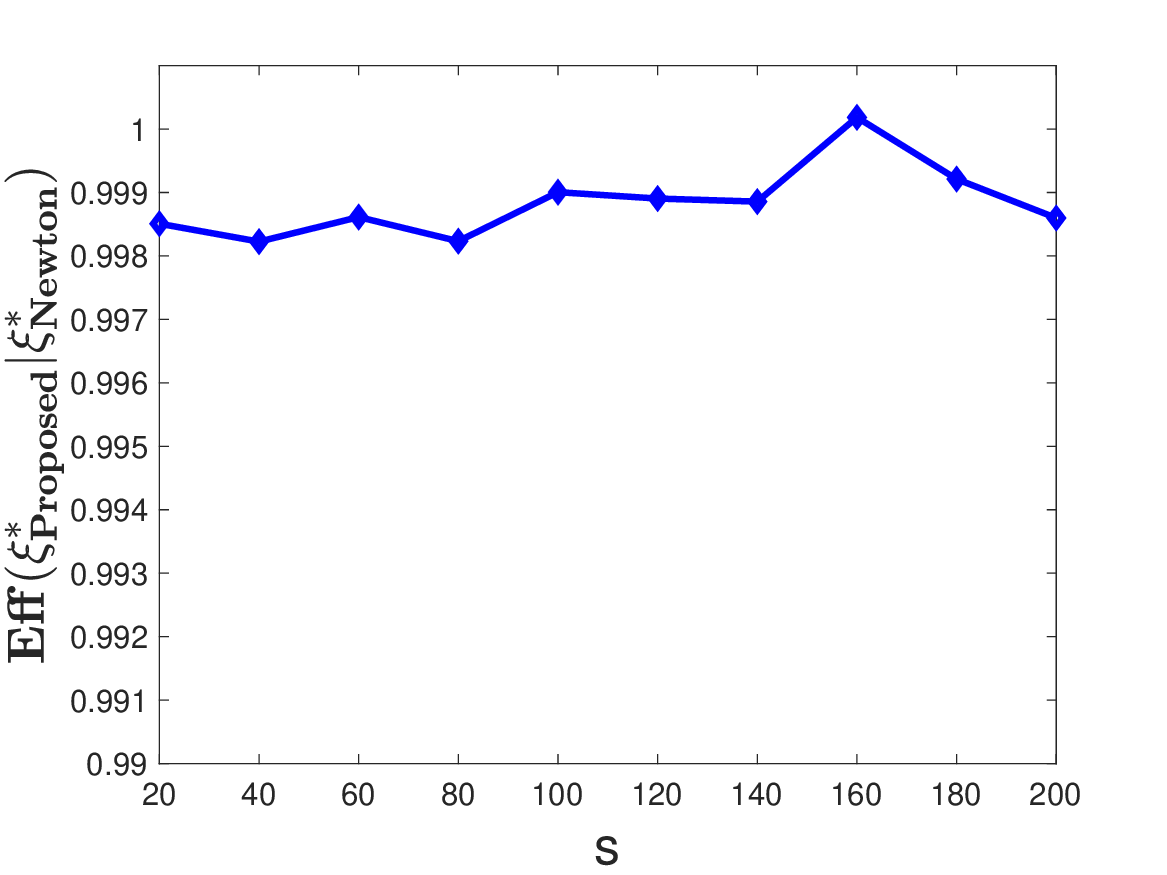}}}
\qquad
\subfloat[$d=2$, Candidate Pool Size $N = (s+1)^2$]
{{\includegraphics[width=6cm]{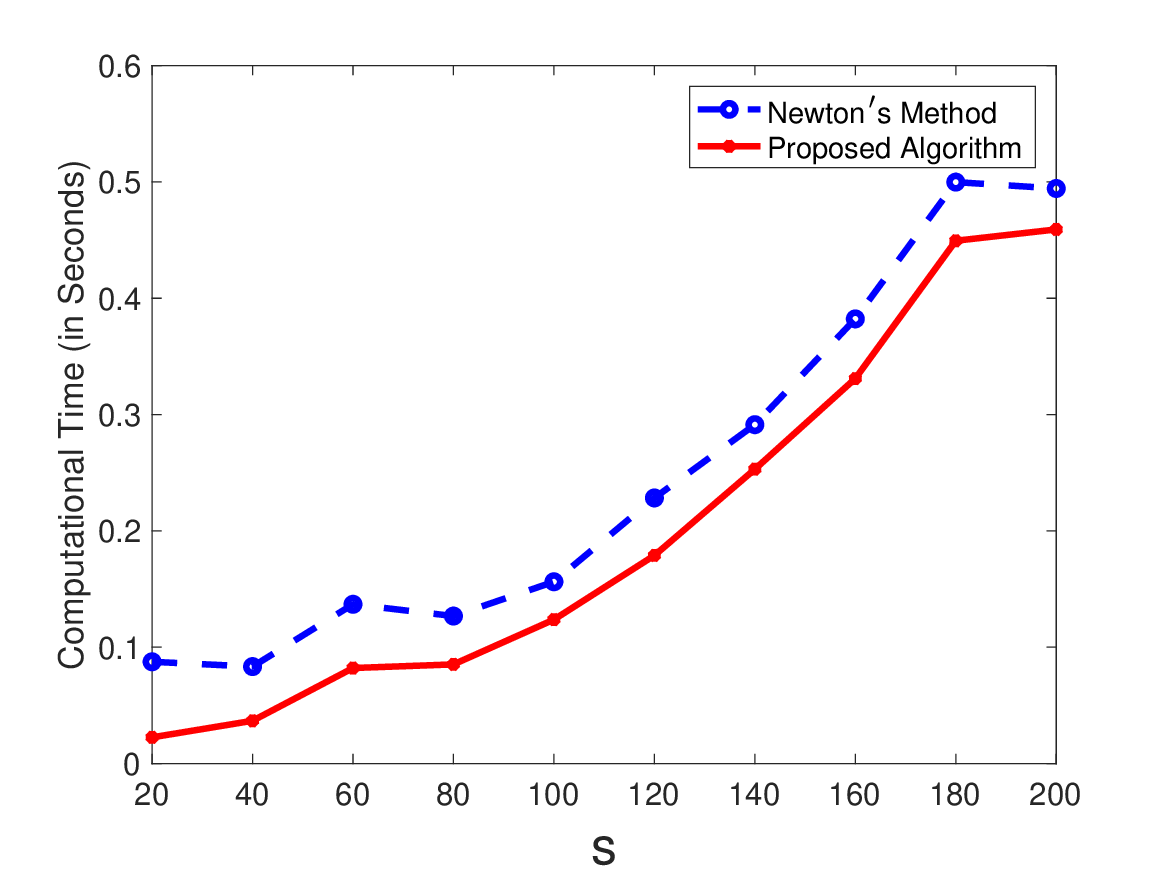}}
\hfill
{\includegraphics[width=6cm]{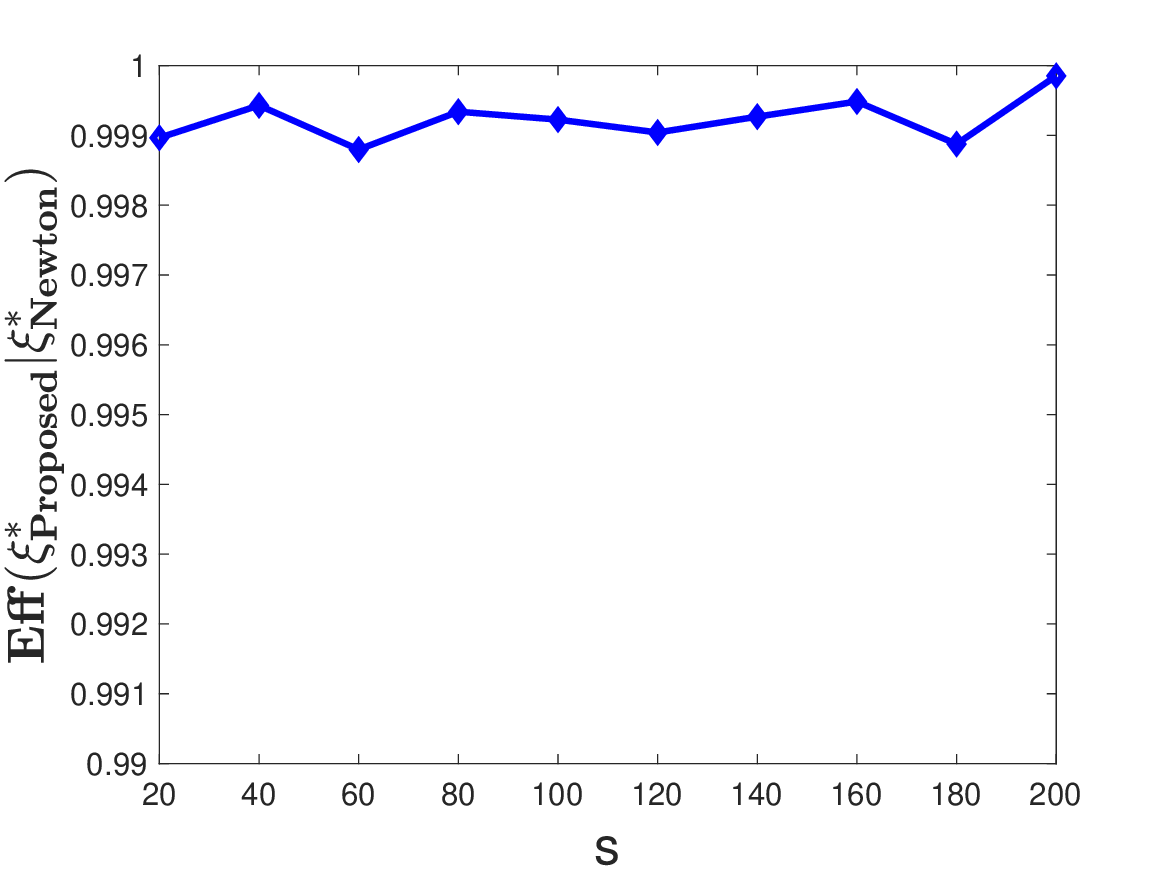}}}
\qquad
\subfloat[$d=3$, Candidate Pool Size $N = (s+1)^3$]
{{\includegraphics[width=6cm]{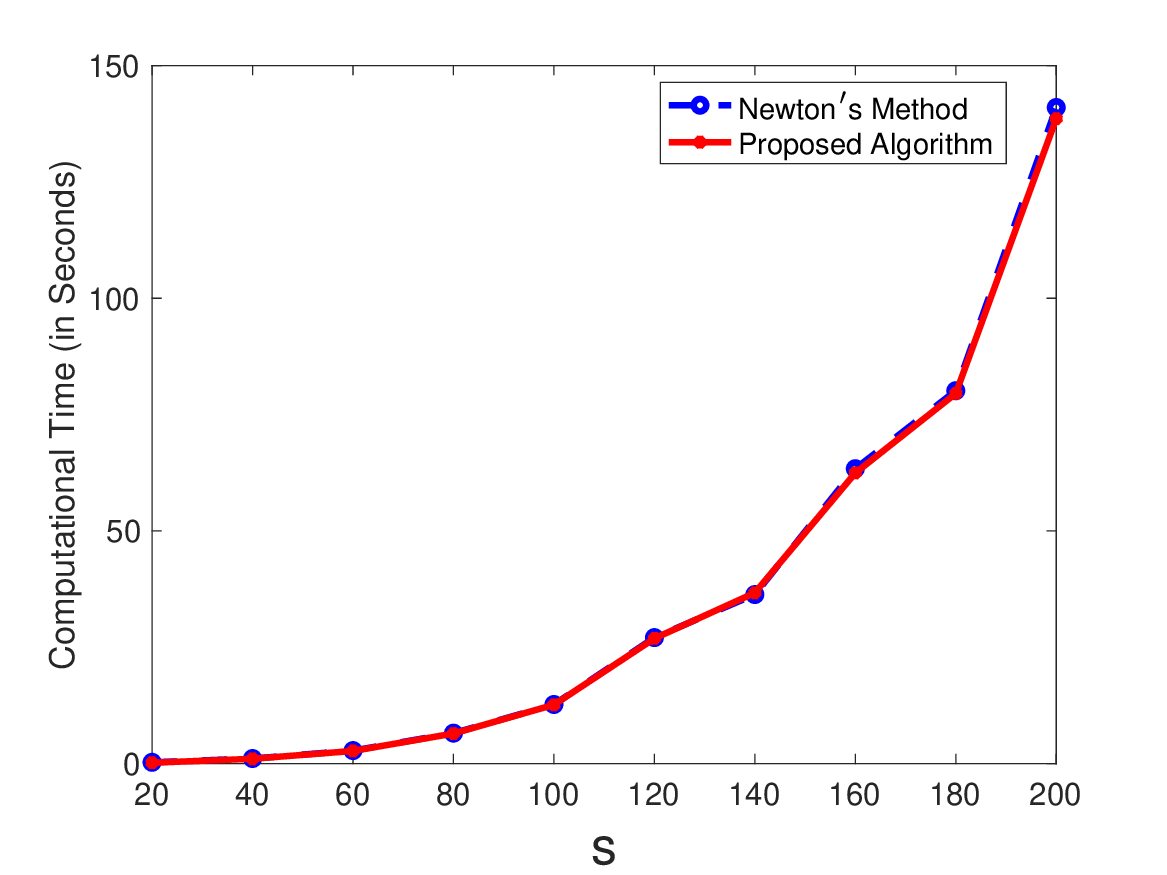}}
\hfill
{\includegraphics[width=6cm]{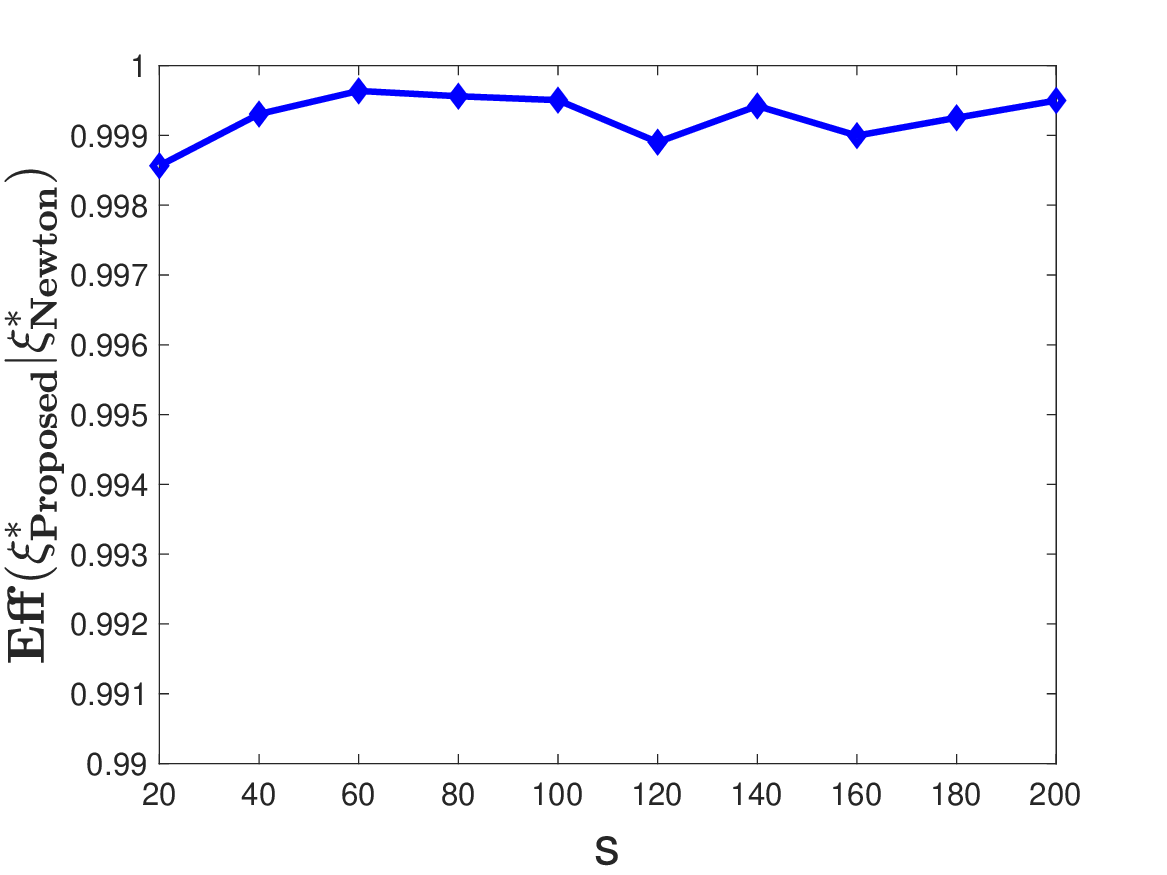}}}
\caption{Computational Time (in Seconds) and Relative Efficiency $\Eff(\xi^*_{\text{Proposed}}|\xi^*_{\text{Newton}})$ for EI-Optimal Designs with $d$-Dimensional Poisson Regression Model}
\label{fig:multiI-eg3}
\end{figure}

\emph{Example 4}. In this example, we would like to provide some comparison between EI-optimal design and other parameter estimation oriented D- and A-optimal designs through a real-world potato packing example in   \citeI{woods2006designs}. The experiment contains $d=3$ quantitative variables - vitamin concentration in the prepackaging dip and levels of two gases in the packing atmosphere.
The response is binary representing the presence or absence of liquid in the pack after 7 days. All explanatory variables are standardized, and the experimental region $\Omega = [-1,1]^3$.
We consider one of the candidate models used in the real study: a logit model with quadratic basis.
The estimates of regression coefficients from preliminary study are given in Table \ref{tab:PotatoPackModel}. $F_{\IMSE} = F_{\unif}$ on $\Omega$ is used to define EI-optimality.

\begin{table}[htbp]
  \centering
  \caption{Logit Model of Potato Packing Example}
    \begin{tabular}{lccccccc}
    \hline
    Term  & Intercept &  $x_2$ & $x_3$ &  $x_2x_3$ &$x_1^2$   & $x_2^2$  &  $x_3^2$\\
    \hline
    Coefficient & -2.93 &  -0.52 & -0.79 & -0.66 & 0.94 & 0.79 & 1.82\\
    \hline
    \end{tabular}%
  \label{tab:PotatoPackModel}%
\end{table}%
Figure \ref{fig:multiI-eg4} shows the support points of D-, A- and EI-Optimal Designs. The relative EI-efficiency of D- and A- optimal designs relative to EI-optimal design are $80.64\%$ and $85.44\%$, respectively. To some extent, this demonstrates the importance and necessity of accurate estimate of regression coefficients to make precise prediction,
but D- or A-optimality may not be the most appropriate criterion to use when the prediction is of interest. The relative A- and D-efficiency of the EI-optimal design relative to the corresponding A- and D-optimal designs are $81.12\%$ and $88.76\%$, respectively. Thus, it is important to choose the appropriate optimality criterion based on the purpose of the experiment.

\begin{figure}[htpb]
\centering
{\includegraphics[width=10cm]{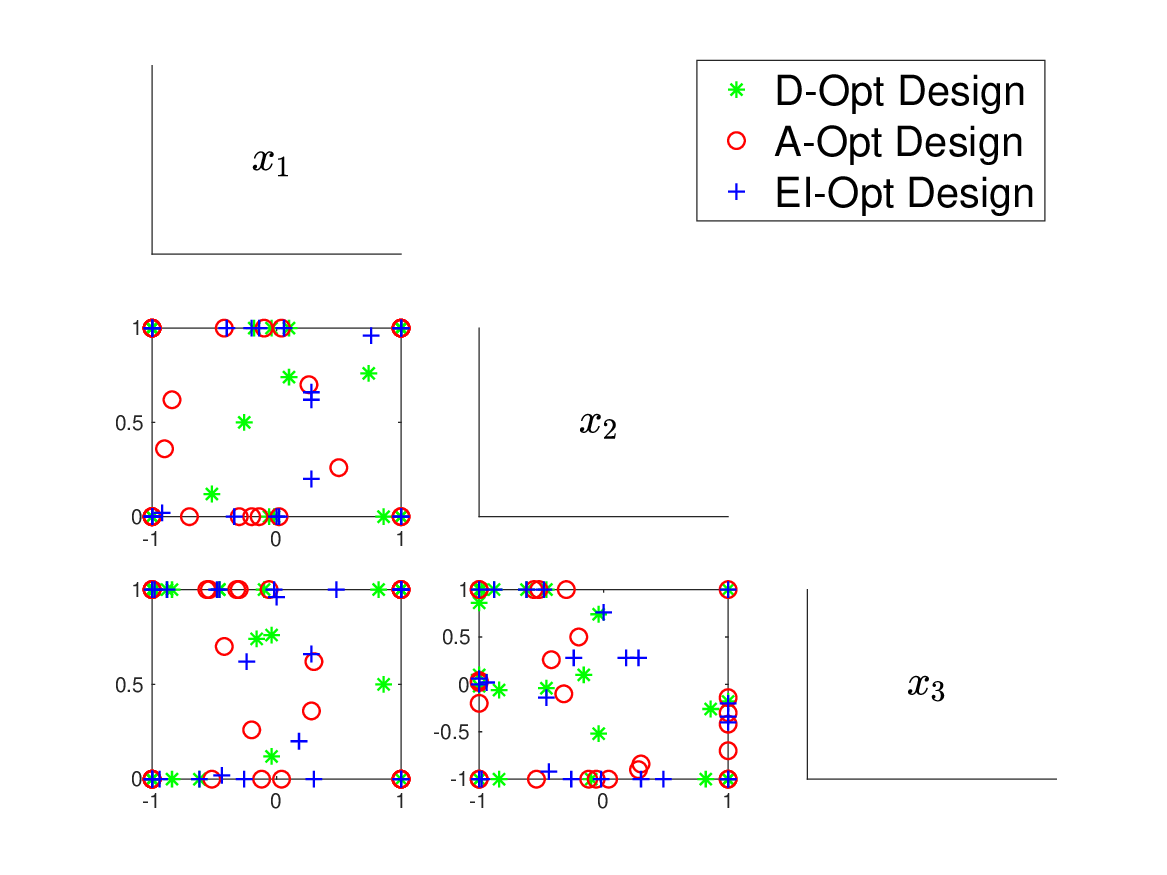}}
\caption{Support Points of D-, A- and EI-Optimal Designs}
\label{fig:multiI-eg4}
\end{figure}

\section{Discussion}\label{sec: discussion}
In this work, we study a general and flexible prediction-oriented criterion EI-optimality for GLMs and advance an efficient sequential algorithm with sound convergence properties for constructing EI-optimal designs.
Through a deep investigation on the theoretical properties of the EI-optimality, we have obtained an insightful understanding for the proposed algorithm on how to sequentially choose the support point and update the weights of support points of the design.
The computational advantages of the proposed algorithm over the Newton's method is demonstrated through numerical examples with moderate size candidate pools for various types of GLMs.
Moreover, all the computations in the proposed algorithm are explicit and simple to implement.
The proposed method exemplifies a good case on the integration of theory and computation to advance the development of new statistical methodology.

It is worth remarking that the proposed sequential algorithm (Algorithm 2) is not restricted to the EI-optimality for constructing optimal designs.
The proposed sequential algorithm can be extended to other optimality criteria when the directional derivative  $\phi(\xi',\xi)$ of optimality criterion in \eqref{eqn:dirder} exists. Although the convergence property of the proposed algorithm (Theorem \ref{thm:cong-algo2}) in Section \ref{section: Algorithm 2} is stated in the context of $\EI$-optimality, a proof on the convergence of the proposed algorithm for general $\Phi_p$-optimality is provided in the appendix. As the $\EI$-optimlity has the same mathematical structure as $\Phi_1$-optimality, the convergence of the algorithm still holds.
Note that this work focuses on the local EI-optimal designs for the generalized linear models, which depends on given regression coefficients.
The parameter dependence problem in the design of GLMs is an important yet challenging issue (\citeO{06khu, woods2006designs}).  For instance, the EI-efficiency of the compromise design proposed by \citeI{woods2006designs} relative to the local EI-optimal design under three model specifications (see \citeO{woods2006designs}) are $49\%$, $80\%$ and $92\%$.
Bayesian or pseudo-Bayesian approaches are generally recognized as an appealing solution to address the parameter dependence issue.
For our proposed method of EI-optimal design, we will further investigate the robust EI-optimal design for the GLMs.
One possibility is to establish a tight upper bound for the integrated mean squared error to relax the dependency on the regression coefficients.
Then we can modify the proposed algorithm in search of optimal designs based on the robust optimality criterion constructed according to the upper bound.
\citeI{hickernell02} developed some theoretical results in a similar direction for linear regression models under model uncertainty, and \citeI{14li} extended the results to linear regression models with gradient information.

We would like to remark that the proposed method may not have great computational efficiency when the candidate pool gets very large.
One direction for tackling this issue is to better take advantage of the information on probability measure in the design region for constructing the design.
For example, one can modify the selection of candidate pool based on the probability measure in the design region.
Another interesting point is that the support points in the optimal design can be outside of the prediction region of interest, which could be a potential limitation of the proposed method.

For future research, it will be interesting to construct EI-optimal designs for the models with both quantitative and qualitative responses (\citeO{deng2015qq}).
Since the GLMs include both the linear regression for continuous response and the logistic regression for binary response,
it would be interesting to study the EI-optimal designs under the consideration of jointly modeling both quantitative and qualitative responses.
Finally, the proposed sequential algorithm is not restricted to find desirable designs for physical experiments.
It can be applied for finding space-filling design in computer experiments (\citeO{deng2015design}).

\bibliographystyle{asa}
\bibliography{GLM}

\section*{Appendix}
{\bf Proof of Lemma \ref{lem:IMSE}}
\begin{proof}
\begin{eqnarray*}
\EI(\xi,\vbeta,F_{\IMSE}) &=& \E\left[\int_{\Omega} \vc(\vx)^T\left(\hat{\vbeta}-\vbeta\right)\left(\hat{\vbeta}-\vbeta\right)^T\vc(\vx)\dif  F_{\IMSE}(\vx)\right]\\
&=& \int_{\Omega} \vc(\vx)^T\E\left[\left(\hat{\vbeta}-\vbeta\right)\left(\hat{\vbeta}-\vbeta\right)^T\right]\vc(\vx)\dif F_{\IMSE}(\vx)\\
&\approx & \int_{\Omega}\tr\left(\vc(\vx)\vc(\vx)^T\mI\left(\xi,\vbeta\right)^{-1}\right)\dif F_{\IMSE}(\vx)\\
&=& \tr\left[\left(\int_{\Omega}\vc(\vx)\vc(\vx)^T\dif F_{\IMSE}(\vx)\right) \mI(\xi,\vbeta)^{-1}\right]\\
&=& \tr\left(\mA \mI(\xi,\vbeta)^{-1}\right).
\end{eqnarray*}
The approximation is provided by the fact that the estimated regression coefficient $\hat{\vbeta}$ follows $N(\vbeta,\mI(\xi, \vbeta)^{-1})$ asymptotically.
\end{proof}

\noindent {\bf Proof of Lemma \ref{lem:ConWeight}}
\begin{proof}
For matrix $\frac{\partial \vf(\vbeta)}{\partial \vbeta^T}$ with size $q\times l$ and matrix $\mI(\xi^{\vlambda})^{-1}\left(\frac{\partial \vf(\vbeta)}{\partial \vbeta^T}\right)^T$ with size $l\times q$, and $q\leq l$, the eigenvalues of $\frac{\partial \vf(\vbeta)}{\partial \vbeta^T}\mI(\xi^{\vlambda})^{-1}\left(\frac{\partial \vf(\vbeta)}{\partial \vbeta^T}\right)^T$  are the eigenvalues of  $\mI(\xi^{\vlambda})^{-1}\left(\frac{\partial \vf(\vbeta)}{\partial \vbeta^T}\right)^T\frac{\partial \vf(\vbeta)}{\partial \vbeta^T}$, with extra eigenvalues being 0 if there is any. Thus, the eigenvalues of$\left[\mI(\xi^{\vlambda})^{-1}\left(\frac{\partial \vf(\vbeta)}{\partial \vbeta^T}\right)^T\frac{\partial \vf(\vbeta)}{\partial \vbeta^T}\right]^p$  are the eigenvalues of $\left[\frac{\partial \vf(\vbeta)}{\partial \vbeta^T}\mI(\xi^{\vlambda})^{-1}\left(\frac{\partial \vf(\vbeta)}{\partial \vbeta^T}\right)^T\right]^p$  with extra eigenvalues being 0. Thus, we could rewrite
\small
\begin{eqnarray}\label{eqn:conv_weight}
\tr\left[\frac{\partial \vf(\vbeta)}{\partial \vbeta^T}\mI(\xi^{\vlambda})^{-1}\left(\frac{\partial \vf(\vbeta)}{\partial \vbeta^T}\right)^T\right]^p = \tr\left[\mI(\xi^{\vlambda})^{-1}\left(\frac{\partial \vf(\vbeta)}{\partial \vbeta^T}\right)^T\left(\frac{\partial \vf(\vbeta)}{\partial \vbeta^T}\right)\right]^p = \tr\left[\mI^{-1}(\xi^{\vlambda})\mB\right]^{p},
\end{eqnarray}
\normalsize
where $\mB = \left(\frac{\partial \vf(\vbeta)}{\partial \vbeta^T}\right)^T\left(\frac{\partial \vf(\vbeta)}{\partial \vbeta^T}\right)$ which is positive semidefinite with size $l\times l$ and rank $q\leq l$. Then, by Smith decomposition, there exists a nonsingular matrix $\mS$ of size $l\times l$ such that
$$
\mB = \mS^T\left(\begin{array}{cc}
\boldsymbol{I}_q & \boldsymbol{0}_{q\times (l-q)}\\
\boldsymbol{0}_{(l-q)\times q} & \boldsymbol{0}_{(l-q)\times (l-q)}
\end{array}\right)\mS = \mS^T
\left(\begin{array}{c}
\boldsymbol{I}_q \\
\boldsymbol{0}_{(l-q)\times q}
\end{array}\right)
\left(\boldsymbol{I}_q\,\,\,\,\, \boldsymbol{0}_{q\times (l-q)}\right)\mS,
$$
where $\boldsymbol{I}_q$ is the identity matrix of size $q\times q$.
Thus, Equation (\ref{eqn:conv_weight}) could be written as
\begin{eqnarray}\label{eqn:conv_weight2}
&&\tr\left[\frac{\partial \vf(\vbeta)}{\partial \vbeta^T}\mI(\xi^{\vlambda})^{-1}\left(\frac{\partial \vf(\vbeta)}{\partial \vbeta^T}\right)^T\right]^p\nonumber\\
 &=& \tr\left[\mI(\xi^{\vlambda})^{-1}\mS^T
\left(\begin{array}{c}
\boldsymbol{I}_q \\
\boldsymbol{0}_{(l-q)\times q}
\end{array}\right)
\left(\boldsymbol{I}_q\,\,\,\,\, \boldsymbol{0}_{q\times (l-q)}\right)\mS\right]^{p}\nonumber\\
& = & \tr\left\{\left[\left(\boldsymbol{I}_q\,\,\,\,\, \boldsymbol{0}_{q\times (l-q)}\right)\left((\mS^T)^{-1}\mI(\xi^{\vlambda})\mS^{-1}\right)^{-1}
\left(\begin{array}{c}
\boldsymbol{I}_q \\
\boldsymbol{0}_{(l-q)\times q}
\end{array}\right)
\right]^{-1}\right\}^{-p}.
\end{eqnarray}
Consider two weights $\vlambda_1$ and $\vlambda_2$, and define $\vlambda_3 = (1-a)\vlambda_1+a\vlambda_2$. With $0\leq a\leq1$, $\vlambda_3$ is still a feasible weight vector, and
\begin{eqnarray}\label{eqn:lin_weight}
(\mS^T)^{-1}\mI(\xi^{\vlambda_3})\mS^{-1} = (1-a)(\mS^T)^{-1}\mI(\xi^{\vlambda_1})\mS^{-1}+a(\mS^T)^{-1}\mI(\xi^{\vlambda_2})\mS^{-1}.
\end{eqnarray}
By the theorem in section 3.13 of \citeI{93puk},
\begin{eqnarray}\label{eq:confunc}
\left[\left(\boldsymbol{I}_q\,\,\,\,\, \boldsymbol{0}_{q\times (l-q)}\right)\left((\mS^T)^{-1}\mI(\xi^{\vlambda})\mS^{-1}\right)^{-1}\left(
\boldsymbol{I}_q\,\,\,\,\, \boldsymbol{0}_{(l-q)\times q}\right)^T
\right]^{-1}
\end{eqnarray}
is matrix concave in $(\mS^T)^{-1}\mI(\xi^{\vlambda})\mS^{-1}$. Thus, by linearity of matrix $(\mS^T)^{-1}\mI(\xi^{\vlambda})\mS^{-1}$ in weight in \eqref{eqn:lin_weight}, $ \left[\left(\boldsymbol{I}_q\,\,\,\,\, \boldsymbol{0}_{q\times (l-q)}\right)\left((\mS^T)^{-1}\mI(\xi^{\vlambda})\mS^{-1}\right)^{-1}\left(
\boldsymbol{I}_q\,\,\,\,\, \boldsymbol{0}_{(l-q)\times q}\right)^T
\right]^{-1}$ is also concave in weight vector $\vlambda$.

Since $\tr(\mC^{-p})^{1/p}$ is nonincreasing and convex for any positive semidefinite matrix $\mC$ (\citeO{fedorov1997model}, see page 22), together with concavity of \eqref{eq:confunc}, the composite function
$$\left[\tr\left\{\left[\left(\boldsymbol{I}_q\,\,\,\,\, \boldsymbol{0}_{q\times (l-q)}\right)\left((\mS^T)^{-1}\mI(\xi^{\vlambda})\mS^{-1}\right)^{-1}\left(
\boldsymbol{I}_q\,\,\,\,\, \boldsymbol{0}_{(l-q)\times q}\right)^T
\right]^{-1}\right\}^{-p}\right]^{1/p}$$
is a convex function of weight vector $\vlambda$ (\citeO{09ber}, see page 480).
As a result, by Equation \eqref{eqn:conv_weight2}, $\Phi_p(\xi^{\vlambda})$ is convex in weight vector $\vlambda$.
\end{proof}

\noindent {\bf Proof of Lemma \ref{lem:DerWeight}}

\begin{proof}
Given $\tilde{\vlambda} = (1-\alpha)\vlambda+\alpha\Delta\vlambda$, we have $\mI(\xi^{\tilde{\vlambda}}) = (1-\alpha)\mI(\xi^{\vlambda})+\alpha\mI(\xi^{\Delta\vlambda})$. We still use Equation (\ref{eqn:conv_weight}) to rewrite $\Phi_p(\xi^{\vlambda})$ as
$$\Phi_p(\xi^{\vlambda}) = \left(q^{-1}\tr\left[\mI(\xi^{\vlambda})^{-1}\mB\right]^{p}\right)^{1/p},$$
where $\mB = \left(\frac{\partial \vf(\vbeta)}{\partial \vbeta^T}\right)^T\left(\frac{\partial \vf(\vbeta)}{\partial \vbeta^T}\right)$ .

For any positive semidefinite matrix $\mC$ as a function of $\alpha$, the derivative of its inverse $\mC^{-1}$ can be calculated as $\frac{\partial \mC^{-1}}{\partial \alpha} = -\mC^{-1}\frac{\partial \mC}{\partial \alpha}\mC^{-1}$ (\citeO{09ber}). So, the derivative of $\mI(\xi^{\tilde{\vlambda}})^{-1}\mB$ with respect to $\alpha$ can be expressed as,
\begin{equation*}\label{eqn:derder}
\frac{\partial\left[\mI(\xi^{\tilde{\vlambda}})^{-1}\mB\right]}{\partial\alpha} = \frac{\partial\mI(\xi^{\tilde{\vlambda}})^{-1})}{\partial\alpha}\mB = -\mI(\xi^{\tilde{\vlambda}})^{-1}[\mI(\xi^{\Delta\vlambda})-\mI(\xi^{\vlambda})]\mI(\xi^{\tilde{\vlambda}})^{-1}\mB.
\end{equation*}
Then, the directional derivative of $\Phi_p(\xi^{\vlambda}) $ is
\begin{eqnarray*}
&&\psi(\Delta{\vlambda},\vlambda)\\
 &=& \left.\frac{\partial \Phi_p(\xi^{\tilde{\vlambda}})}{\partial\alpha}\right|_{\alpha=0}\\
&=& \left.q^{-1/p}\left[\tr\left(\mI(\xi^{\tilde{\vlambda}})^{-1}\mB\right)^p\right]^{1/p-1}\tr\left[\left(\mI(\xi^{\tilde{\vlambda}})^{-1}\mB\right)^{p-1}\left( -\mI(\xi^{\tilde{\vlambda}})^{-1}[\mI(\xi^{\Delta\vlambda})-\mI(\xi^{\vlambda})]\mI(\xi^{\tilde{\vlambda}})^{-1}\mB\right)\right]\right|_{\alpha=0}\\
&=& q^{-1/p}\left[\tr\left(\mI(\xi^{\vlambda})^{-1}\mB\right)^p\right]^{1/p-1}\tr\left[\left(\mI(\xi^{\vlambda})^{-1}\mB\right)^p -\left(\mI(\xi^{\vlambda})^{-1}\mB\right)^{p-1} \mI(\xi^{\vlambda})^{-1}\mI(\xi^{\Delta\vlambda})\mI(\xi^{\vlambda})^{-1}\mB\right]\\
&=& \Phi_p(\xi^{\vlambda})-q^{-1/p}\left[\tr\left(\mI(\xi^{\vlambda})^{-1}\mB\right)^p\right]^{1/p-1}\tr\left[\left(\mI(\xi^{\vlambda})^{-1}\mB\right)^{p-1} \mI(\xi^{\vlambda})^{-1}\mI(\xi^{\Delta\vlambda})\mI(\xi^{\vlambda})^{-1}\mB\right].
\end{eqnarray*}
\end{proof}

\noindent {\bf Proof of Theorem \ref{thm:OptWeights}}
\begin{proof}
Since $\vlambda^*$ minimizes $\Phi_p$, and $\Phi_p$ is a convex function of weight vector as proved in Lemma \ref{lem:ConWeight},
\begin{eqnarray*}
&&\psi(\Delta{\vlambda}, \vlambda^*)\\
 &=& \Phi_p(\xi^{\vlambda^*})-q^{-1/p}\left[\tr\left(\mI(\xi^{\vlambda})^{-1}\mB\right)^p\right]^{1/p-1}\tr\left[\left(\mI(\xi^{\vlambda^*})^{-1}\mB\right)^{p-1} \mI(\xi^{\vlambda})^{-1}\mI(\xi^{\Delta\vlambda})\mI(\xi^{\vlambda})^{-1}\mB\right]\\
&=& \Phi_p(\xi^{\vlambda^*}) \\
&& \hspace{-2.8ex}-\sum\limits_i^n\Delta\lambda_i q^{-1/p}\left[\tr\left(\mI(\xi^{\vlambda^*})^{-1}\mB\right)^p\right]^{1/p-1}w(\vx_i)\vg(\vx_i)^T\mI(\xi^{\vlambda^*})^{-1}\mB\left(\mI(\xi^{\vlambda^*})^{-1}\mB\right)^{p-1}\mI(\xi^{\vlambda^*})^{-1}\vg(\vx_i)\\
&=& \sum\limits_{i=1}^n\Delta{\lambda}_i\left[\Phi_p(\xi^{\vlambda^*})\right.\\
&&\hspace{-2.8ex}\left.-q^{-1/p}\left[\tr\left(\mI(\xi^{\vlambda^*})^{-1}\mB\right)^p\right]^{1/p-1}w(\vx_i)\vg(\vx_i)^T\mI(\xi^{\vlambda^*})^{-1}\mB\left(\mI(\xi^{\vlambda^*})^{-1}\mB\right)^{p-1}\mI(\xi^{\vlambda^*})^{-1}\vg(\vx_i)\right]
\geq 0,
\end{eqnarray*}
for all feasible weight vector $\Delta{\vlambda}$.

Thus,
$$\Phi_p(\xi^{\vlambda^*})\geq q^{-1/p}\left[\tr\left(\mI(\xi^{\vlambda^*})^{-1}\mB\right)^p\right]^{1/p-1}w(\vx_i)\vg(\vx_i)^T\mI(\xi^{\vlambda^*})^{-1}\mB\left(\mI(\xi^{\vlambda^*})^{-1}\mB\right)^{p-1}\mI(\xi^{\vlambda^*})^{-1}\vg(\vx_i),$$ for $i=1,...,n$.

Now we will show,
$$\Phi_p(\xi^{\vlambda^*}) = q^{-1/p}\left[\tr\left(\mI(\xi^{\vlambda^*})^{-1}\mB\right)^p\right]^{1/p-1}w(\vx_i)\vg(\vx_i)^T\mI(\xi^{\vlambda^*})^{-1}\mB\left(\mI(\xi^{\vlambda^*})^{-1}\mB\right)^{p-1}\mI(\xi^{\vlambda^*})^{-1}\vg(\vx_i),$$
for all $\lambda^*_i>0$.

Suppose there exists at least one $\vx_j$ with $\lambda_j^*>0$, such that
$$\Phi_p(\xi^{\vlambda^*})>q^{-1/p}\left[\tr\left(\mI(\xi^{\vlambda^*})^{-1}\mB\right)^p\right]^{1/p-1}w(\vx_j)\vg(\vx_j)^T\mI(\xi^{\vlambda^*})^{-1}\mB\left(\mI(\xi^{\vlambda^*})^{-1}\mB\right)^{p-1}\mI(\xi^{\vlambda^*})^{-1}\vg(\vx_j).$$
Then, we have
\begin{eqnarray*}
&&\Phi_p(\xi^{\vlambda^*}) \\
&=&  \sum\limits_{i=1}^n\lambda^*_i\Phi_p(\xi^{\vlambda^*})\\
&>& \sum\limits_{i=1}^n\lambda^*_iq^{-1/p}\left[\tr\left(\mI(\xi^{\vlambda^*})^{-1}\mB\right)^p\right]^{1/p-1}w(\vx_i)\vg(\vx_i)^T\mI(\xi^{\vlambda^*})^{-1}\mB\left(\mI(\xi^{\vlambda^*})^{-1}\mB\right)^{p-1}\mI(\xi^{\vlambda^*})^{-1}\vg(\vx_i) \\
&=& q^{-1/p}\left[\tr\left(\mI(\xi^{\vlambda^*})^{-1}\mB\right)^p\right]^{1/p-1}\sum\limits_{i=1}^n\lambda^*_iw(\vx_i)\vg(\vx_i)^T\mI(\xi^{\vlambda^*})^{-1}\mB\left(\mI(\xi^{\vlambda^*})^{-1}\mB\right)^{p-1}\mI(\xi^{\vlambda^*})^{-1}\vg(\vx_i)\\
&=& q^{-1/p}\left[\tr\left(\mI(\xi^{\vlambda^*})^{-1}\mB\right)^p\right]^{1/p-1}\tr\left(\mI(\xi^{\vlambda^*})^{-1}\mB\right)^p\\
&=& \Phi_p(\xi^{\vlambda^*}),
\end{eqnarray*}
which is a contradiction.
So,
$$\Phi_p(\xi^{\vlambda^*}) = q^{-1/p}\left[\tr\left(\mI(\xi^{\vlambda^*})^{-1}\mB\right)^p\right]^{1/p-1}w(\vx_i)\vg(\vx_i)^T\mI(\xi^{\vlambda^*})^{-1}\mB\left(\mI(\xi^{\vlambda^*})^{-1}\mB\right)^{p-1}\mI(\xi^{\vlambda^*})^{-1}\vg(\vx_i),$$
for design points $\vx_i$ with $\lambda^*_i>0$.
\end{proof}

\noindent {\bf Proof of Corollary \ref{cor:I-weight_der}}
\begin{proof}
This is a special case that $\mB = q\mA = q\int_{\Omega} \vg(\vx)\vg^T(\vx)\left[\frac{\dif h^{-1}}{\dif \eta}\right]^2\dif F_{\IMSE}(\vx)$ and $p=1$. Thus,
\begin{eqnarray*}
\psi(\Delta{\vlambda},\vlambda) &=&  \tr\left(\mI(\xi^{\vlambda})^{-1}\mA\right)-\tr\left[ \mI(\xi^{\vlambda})^{-1}\mI(\xi^{\Delta\vlambda})\mI(\xi^{\vlambda})^{-1}\mA\right]\\
&=& \tr\left(\mI(\xi^{\vlambda})^{-1}\mA\right)-\tr\left[ \mI(\xi^{\vlambda})^{-1}\left(\sum\limits_{i}\Delta\lambda_iw(\vx_i)\vg(\vx_i)\vg(\vx_i)^T\right)\mI(\xi^{\vlambda})^{-1}\mA\right]\\
&=& \tr\left(\mI(\xi^{\vlambda})^{-1}\mA\right)-\sum\limits_{i=1}^n\Delta{\lambda}_iw(\vx_i)\vg(\vx_i)^T\mI(\xi^{\vlambda})^{-1}\mA\mI(\xi^{\vlambda})^{-1}\vg(\vx_i)
\end{eqnarray*}
\end{proof}

\noindent {\bf Proof of Corollary \ref{cor:suf_weight}}
\begin{proof}
For (i), it directly follows the result in Theorem \ref{thm:OptWeights}. For (ii), the proof follows exactly as the proof of Corollary \ref{cor:I-weight_der}.
\end{proof}

\noindent {\bf Proof of Corollary \ref{cor: feasibility}}
\begin{proof}
A positive definite $\mA$ in EI-optimality would insure $\mI(\xi^{\vlambda^k})^{-1}\mB\neq 0$ in Theorem \ref{thm:feasibility} with $\mB=q\mA$ and $p=1$.
\end{proof}

\noindent {\bf Proof of Theorem \ref{thm:feasibility}}
\begin{proof} Since information matrix $\mI(\xi^{\vlambda^k})$ is positive definite, $\mB$ is positive semidefinite, and $\mI(\xi^{\vlambda^k})^{-1}\mB\neq 0$, we have
\begin{eqnarray*}
0<\tr\left[\mI(\xi^{\vlambda^k})^{-1}\mB\right]^p&=&\tr\left[\mI(\xi^{\vlambda^k})^{-1}\mB\left[\mI(\xi^{\vlambda^k})^{-1}\mB\right]^{p-1}\mI(\xi^{\vlambda^k})^{-1}\mI(\xi^{\vlambda^k})\right]\\
&=& \tr\left[\sum\limits_{i=1}^n\lambda^k_iw(\vx_i)\vg(\vx_i)\vg(\vx_i)^T\mI(\xi^{\vlambda^k})^{-1}\mB\left[\mI(\xi^{\vlambda^k})^{-1}\mB\right]^{p-1}\mI(\xi^{\vlambda^k})^{-1}\right]\\
&=&  \sum\limits_{i=1}^n\lambda^k_i\tr\left[w(\vx_i)\vg(\vx_i)\vg(\vx_i)^T\mI(\xi^{\vlambda^k})^{-1}\mB\left[\mI(\xi^{\vlambda^k})^{-1}\mB\right]^{p-1}\mI(\xi^{\vlambda^k})^{-1}\right]\\
&=& \sum\limits_{i=1}^n\lambda^k_iw(\vx_i)\vg(\vx_i)^T\mI(\xi^{\vlambda^k})^{-1}\mB\left[\mI(\xi^{\vlambda^k})^{-1}\mB\right]^{p-1}\mI(\xi^{\vlambda^k})^{-1}\vg(\vx_i).
\end{eqnarray*}
Thus, there exists some $\vx_i$, such that
$$\lambda^k_iw(\vx_i)\vg(\vx_i)^T\mI(\xi^{\vlambda^k})^{-1}\mB\left[\mI(\xi^{\vlambda^k})^{-1}\mB\right]^{p-1}\mI(\xi^{\vlambda^k})^{-1}\vg(\vx_i)>0,$$ and naturally, $\lambda^k_i\left(w(\vx_i)\vg(\vx_i)^T\mI(\xi^{\vlambda^k})^{-1}\mA\mI(\xi^{\vlambda^k})^{-1}\vg(\vx_i)\right)^{\delta}>0$, which leads to
$$\sum\limits_{i=1}^n \lambda_i^k\left[w(\vx_i)\vg(\vx_i)^T\mI(\xi^{\vlambda^k})^{-1}\mB\left(\mI(\xi^{\vlambda^k})^{-1}\mB\right)^{p-1}\mI(\xi^{\vlambda^k})^{-1}\vg(\vx_i)\right]^\delta>0.$$
\end{proof}

\begin{lem}\label{lem:ConvexVarphi}
Define $\varphi(\mI) = \phi_p(\xi) =  \left(q^{-1}\tr\left[\frac{\partial \vf(\vbeta)}{\partial \vbeta^T}\mI^{-1}\left(\frac{\partial \vf(\vbeta)}{\partial \vbeta^T}\right)^T\right]^p\right)^{1/p},\,\,\,0<p<\infty$ as a function of Fisher information matrix $\mI$, then
\begin{itemize}
  \item[(a)]  $\varphi(\mI)$ is a strictly convex function of $\mI$.
  \item[(b)] $\varphi(\mI)$ is a decreasing function of $\mI$.
\end{itemize}
\end{lem}

\begin{proof}
The proof of (a) is very similar to that of Lemma \ref{lem:ConWeight}, and is omitted here.
For (b), for any $\mI_1\preceq\mI_2$,
$$\mI_1^{-1}\succeq\mI_2^{-1}\Rightarrow \frac{\partial \vf(\vbeta)}{\partial \vbeta^T}\mI_1^{-1}\left(\frac{\partial \vf(\vbeta)}{\partial \vbeta^T}\right)^T\succeq \frac{\partial \vf(\vbeta)}{\partial \vbeta^T}\mI_2^{-1}\left(\frac{\partial \vf(\vbeta)}{\partial \vbeta^T}\right)^T.$$
As $\tr(\mC^r)$ with $r>0$ is an increasing function of any positive definite matrix $\mC$ (\citeO{09ber}),
$$\tr\left[\frac{\partial \vf(\vbeta)}{\partial \vbeta^T}\mI_1^{-1}\left(\frac{\partial \vf(\vbeta)}{\partial \vbeta^T}\right)^T\right]^p\geq \tr\left[\frac{\partial \vf(\vbeta)}{\partial \vbeta^T}\mI_2^{-1}\left(\frac{\partial \vf(\vbeta)}{\partial \vbeta^T}\right)^T\right]^p.$$
As a result,
\begin{eqnarray*}
&& \left(q^{-1}\tr\left[\frac{\partial \vf(\vbeta)}{\partial \vbeta^T}\mI_1^{-1}\left(\frac{\partial \vf(\vbeta)}{\partial \vbeta^T}\right)^T\right]^p\right)^{1/p}\geq  \left(q^{-1}\tr\left[\frac{\partial \vf(\vbeta)}{\partial \vbeta^T}\mI_2^{-1}\left(\frac{\partial \vf(\vbeta)}{\partial \vbeta^T}\right)^T\right]^p\right)^{1/p},
\end{eqnarray*}
i.e., $\varphi(\mI_1)\geq \varphi(\mI_2).$
\end{proof}

\noindent {\bf Proof of Proposition \ref{pro:convAlg1}}
\begin{proof}
The proof of Proposition \ref{pro:convAlg1} will mainly based on the results in Theorems 1 and 2 in the paper by \citeI{10yu}.
Based on Lemma \ref{lem:ConvexVarphi}, it is known that
$$\varphi(\mI) = \ \left(q^{-1}\tr\left[\frac{\partial \vf(\vbeta)}{\partial \vbeta^T}\mI^{-1}\left(\frac{\partial \vf(\vbeta)}{\partial \vbeta^T}\right)^T\right]^p\right)^{1/p},\,\,\,0<p<\infty$$ is a convex and decreasing function of $\mI$.
Under Algorithm 1 with $0<\delta< 1$, denote $\vlambda^k$ and $\vlambda^{k+1}$ to be the solutions on $k$th and $(k+1)$th iteration of the multiplicative algorithm, respectively.
Since $\mI(\xi^{\vlambda^k}), \mI(\xi^{\vlambda^{k+1}})$ and $\mB$ are positive definite, it is easy to see that
$$\mI(\xi^{\vlambda^k})>0,\,\,\,\mI(\xi^{\vlambda^{k+1}})>0,\,\,\,\text{and}\,\,\,\mI(\xi^{\vlambda^k})^{-1}\mB\left(\mI(\xi^{\vlambda^k})^{-1}\mB\right)^{p-1}\mI(\xi^{\vlambda^k})^{-1} \neq 0,$$
where $\mI(\xi^{\vlambda^k})^{-1}\mB\left(\mI(\xi^{\vlambda^k})^{-1}\mB\right)^{p-1}\mI(\xi^{\vlambda^k})^{-1}$ is a continuous function for $\mI$. Based on Theorem 1 in \citeI{10yu}, we can have
$$ \phi_p(\xi^{\vlambda^{k+1}}) \le \phi_p(\xi^{\vlambda^{k}}),  $$
when $\vlambda^{k+1} \ne \vlambda^{k}$.
Thus it shows the monotonicity of Algorithm 1.

Moreover, we also have that for $\varphi(\mI) = \phi_p(\xi) =  \left(q^{-1}\tr\left[\frac{\partial \vf(\vbeta)}{\partial \vbeta^T}\mI^{-1}\left(\frac{\partial \vf(\vbeta)}{\partial \vbeta^T}\right)^T\right]^p\right)^{1/p}$,
$$
w(\vx_i)\vg(\vx_i)\vg^T(\vx_i) \frac{\partial \varphi(\mI)}{\partial \mI} |_{\mI = \mI(\xi^{\vlambda^{k}}) } =  w(\vx_i)\vg(\vx_i)\vg^T(\vx_i) \mI(\xi^{\vlambda^k})^{-1}\mB\left(\mI(\xi^{\vlambda^k})^{-1}\mB\right)^{p-1}\mI(\xi^{\vlambda^k})^{-1} \ne 0.
$$
For the sequence $\mI(\xi^{\vlambda^{k}}), k = 1, 2, ...$ from Algorithm 1, its limit point is obviously nonsingular because of the positive definiteness.
Combining the above statements with results in Lemma \ref{lem:ConvexVarphi}, all the required conditions in Theorem 2 of \citeI{10yu} are satisfied.
Thus, using the results in Theorem 2 of \citeI{10yu}, we have
all limit points of $\vlambda^{k}$ are global minims of $\phi_p(\xi^{\vlambda})$, and
the $\phi_p(\xi^{\vlambda})$ decreases monotonically to $\inf\limits_{\vlambda}\phi_p(\xi^{\vlambda})$ as $k\rightarrow\infty$.
\end{proof}

\noindent {\bf Proof of Theorem \ref{thm:efflowerbound}}

To prove Theorem \ref{thm:efflowerbound} and \ref{thm:cong-algo2}, we first prove the following Lemma.
\begin{lem}\label{lem:forproof}
For any design $\xi$ and $\Phi_p$-optimal design $\xi^*$ that minimizes $\Phi_p(\xi)$, the following inequality holds:
$$\min_{\vx\in\Omega}\phi(\vx,\xi)\leq \phi(\xi^*,\xi) \leq \Phi_p(\xi^*)-\Phi_p(\xi)\leq 0,$$
where $\phi(\vx,\xi)$ and $\phi(\xi^*,\xi)$ are the directional derivatives of $\Phi_p(\xi)$ in the direction of $\vx$ and $\xi^*$, respectively.

\begin{proof}
The directional derivative of $\Phi_p(\xi)$ in the direction of the point $\vx^{*} = \argmin\limits_{\vx\in\Omega}\phi(\vx,\xi)$ is:
\begin{eqnarray}\label{inequ:phixxi}
\min\limits_{\vx\in\Omega} \phi(\vx,\xi)&=&\phi(\vx^{*},\xi)\nonumber\\
 &=& \Phi_p(\xi)-q^{-1/p}\left[\tr\left(\mI(\xi)^{-1}\mB\right)^p\right]^{1/p-1}w(\vx^*)\vg(\vx^*)^{\top}\mI(\xi)^{-1}\mB\left(\mI(\xi)^{-1}\mB\right)^{p-1}\mI(\xi)^{-1}\vg(\vx^*)\nonumber\\
 &\leq & \Phi_p(\xi)-q^{-1/p}\left[\tr\left(\mI(\xi)^{-1}\mB\right)^p\right]^{1/p-1}w(\vx)\vg(\vx)^{\top}\mI(\xi)^{-1}\mB\left(\mI(\xi)^{-1}\mB\right)^{p-1}\mI(\xi)^{-1}\vg(\vx)\nonumber\\
\end{eqnarray}
for any $\vx\in\Omega$.

Denote the optimal design $\xi^* = \left\{\begin{array}{ccc}
\vx_1,&...,&\vx_n\\
\lambda^*_1,&...,&\lambda^*_n
\end{array}\right\}$. With \eqref{inequ:phixxi}, we have
\begin{eqnarray}\label{inequ:phixixi}
&&\phi(\vx^*,\xi)\nonumber\\
 &\leq & \sum_{i=1}^n\lambda^*_i\left(\Phi_p(\xi)-q^{-1/p}\left[\tr\left(\mI(\xi)^{-1}\mB\right)^p\right]^{1/p-1}w(\vx_i)\vg(\vx_i)^{\top}\mI(\xi)^{-1}\mB\left(\mI(\xi)^{-1}\mB\right)^{p-1}\mI(\xi)^{-1}\vg(\vx_i)\right)\nonumber\\
 &=&\Phi_p(\xi)-\tr\left[\mI(\xi^*) q^{-1/p}\left[\tr\left(\mI(\xi)^{-1}\mB\right)^p\right]^{1/p-1}\mI(\xi)^{-1}\mB\left(\mI(\xi)^{-1}\mB\right)^{p-1}\mI(\xi)^{-1}\right]\nonumber\\
 &=& \phi(\xi^*,\xi).
\end{eqnarray}

Furthermore, with the definition of directional derivative in the direction of the optimal design $\xi^*$ and convexity of $\Phi_p$, we have
\begin{eqnarray}\label{eqn:dirinequ}
\phi(\xi^*,\xi) &=& \lim_{\alpha\rightarrow 0}\frac{\Phi_p((1-\alpha)\xi+\alpha\xi^*)-\Phi_p(\xi)}{\alpha}\nonumber\\
&\leq &\lim_{\alpha\rightarrow 0}\frac{(1-\alpha)\Phi_p(\xi)+\alpha\Phi_p(\xi^*)-\Phi_p(\xi)}{\alpha}\nonumber\\
&=& \Phi_p(\xi^*)-\Phi_p(\xi)
\end{eqnarray}

Combining \eqref{inequ:phixixi} and \eqref{eqn:dirinequ}, we complete the proof that
$$\min\limits_{\vx\in\Omega} \phi(\vx,\xi)\leq \phi(\xi^*,\xi)\leq \Phi_p(\xi^*)-\Phi_p(\xi)\leq 0.$$
\end{proof}
\end{lem}

\begin{proof} [Proof\nopunct]\emph{of Theorem \ref{thm:efflowerbound}}

In Equation \eqref{eq:phi_p}, the reciprocal of $\Phi_p$-optimality coule be written as:
$$\frac{1}{\Phi_p(\xi)} = \left(q^{-1}\tr\left[\frac{\partial \vf(\vbeta)}{\partial \vbeta^T}\mI(\xi)^{-1}\left(\frac{\partial \vf(\vbeta)}{\partial \vbeta^T}\right)^T\right]^{p}\right)^{-1/p},\,\,\,0<p<\infty.$$
According to \cite{93puk}, Section 6.13, $\frac{1}{\Phi_p}$ is concave on the set of designs with positive definite information matrices.  Using the similar approach to derive the lower bound of efficiency in (\citeO{atwood1969optimal}), with $\xi$ and $\xi^*$ fixed, define $t(\alpha) = \frac{1}{\Phi_p(\alpha\xi^*+(1-\alpha)\xi)}$. With simple algebra and directional derivative of $\Phi_p$,
\begin{eqnarray*}
\left.\frac{\dif t(\alpha)}{\dif \alpha}\right|_{\alpha=0} &=& \left.-\frac{1}{\Phi_p^2(\alpha\xi^*+(1-\alpha)\xi)}\left(\frac{\partial \Phi_p(\alpha\xi^*+(1-\alpha)\xi)}{\partial \alpha}\right)\right|_{\alpha=0}\\
&=& -\frac{1}{\Phi_p^2(\xi)}\phi(\xi^*,\xi),
\end{eqnarray*}
where $\phi(\xi^*,\xi)$ is the directional derivative of $\Phi_p$ in the direction of $\xi^*$ defined in section \ref{section:GET}.

According to Lemma \ref{lem:forproof} that $\phi(\vx^*,\xi)\leq \phi(\xi^*,\xi) \leq \Phi_p(\xi^*)-\Phi_p(\xi)\leq 0,$ we have
\begin{eqnarray*}
&&\frac{1}{\Phi_p(\xi^*)}-\frac{1}{\Phi_p(\xi)}\nonumber\\
 &=&  t(1)-t(0)\nonumber\\
&\leq & \frac{\dif t(\alpha)}{\dif \alpha}|_{\alpha = \alpha^*},\,\,\,\,\,\text{where $\alpha^*\in (0,1)$ by mean value theorem}\nonumber\\
&\leq & \frac{\dif t(\alpha)}{\dif \alpha}|_{\alpha = 0},\,\,\,\,\text{by concavity of $t$}\nonumber\\
&=& -\frac{1}{\Phi_p^2(\xi)}\phi(\xi^*,\xi)\nonumber\\
&\leq & -\frac{1}{\Phi_p^2(\xi)}\phi(\vx^*,\xi)\nonumber\\
&=& \frac{\max\limits_{\vx\in\mathcal{C}} q^{-1/p}\left[\tr\left(\mI(\xi)^{-1}\mB\right)^p\right]^{1/p-1}w(\vx_i)\vg(\vx_i)^T\mI(\xi)^{-1}\mB\left(\mI(\xi)^{-1}\mB\right)^{p-1}\mI(\xi)^{-1}\vg(\vx_i)-\Phi_p(\xi)}{\Phi_p^2(\xi)}\nonumber\\
\end{eqnarray*}
That is, $\frac{\Phi_p(\xi)}{\Phi_p(\xi^*)}\leq \frac{\max\limits_{\vx\in\mathcal{C}} q^{-1/p}\left[\tr\left(\mI(\xi)^{-1}\mB\right)^p\right]^{1/p-1}w(\vx_i)\vg(\vx_i)^T\mI(\xi)^{-1}\mB\left(\mI(\xi)^{-1}\mB\right)^{p-1}\mI(\xi)^{-1}\vg(\vx_i)}{\Phi_p(\xi)},$ i.e.,
$$\Eff_{\Phi_p}(\xi|\xi^*) \geq \frac{\Phi_p(\xi)}{\max\limits_{\vx\in\mathcal{C}}q^{-1/p}\left[\tr\left(\mI(\xi)^{-1}\mB\right)^p\right]^{1/p-1}w(\vx_i)\vg(\vx_i)^T\mI(\xi)^{-1}\mB\left(\mI(\xi)^{-1}\mB\right)^{p-1}\mI(\xi)^{-1}\vg(\vx_i)}.$$
\end{proof}

\noindent {\bf Proof of Theorem \ref{thm:cong-algo2}}

\begin{proof}
Here we prove the convergence that
$$\lim\limits_{r\rightarrow\infty} \Phi_p(\xi^r) = \Phi_p(\xi^*)$$
 of the proposed algorithm for a general $\Phi_p$ optimality, and $\EI$-optimality shares the same mathematical structure as $\Phi_1$-optimality.

We will establish the argument by proof of contradiction with similar proof in the paper of \citeI{13yang}.
Assume that the $\xi^r$ does not converge to $\xi^*$, i.e.,
$$\lim\limits_{r\rightarrow\infty} \Phi_p(\xi^r) -\Phi_p(\xi^*) > 0.$$
Since the support set of the $r$-th iteration is a subset of support set of the ($r+1$)-th iteration, it is obvious that $\Phi_p(\xi^{r+1})\leq \Phi_p(\xi^r)$ for all $r\geq 0$. Thus, there exists some $a>0$, such that
$$\Phi_p(\xi^r) > \Phi_p(\xi^*)+a,\,\,\,\text{for all } r.$$

According to Lemma \ref{lem:forproof}, we can conclude $\phi(\vx^*,\xi^r)\leq \phi(\xi^*,\xi^r) \leq \Phi_p(\xi^*)-\Phi_p(\xi^r)<-a.$ It is quite obvious from the derivation of the directional derivative that $\Phi_p((1-\gamma)\xi+\gamma\xi')$ is infinitely differentiable with respect to $\gamma \in[0,1]$.
Thus, the second order derivative is bounded on $\gamma$, and denote the upper bound by $U$, with $U>0$ as $\EI$ is a convex function.

Then, consider the design $\tilde{\xi}^{r+1} = (1-\gamma)\xi^r+\gamma\vx_r^*$, with $\gamma\in[0,1]$. Since the proposed algorithm achieves optimal weights in each iteration, we have,
$$\Phi_p(\xi^{r+1})\leq \Phi_p(\tilde{\xi}^{r+1}).$$
Using the Taylor expansion of $\Phi_p(\tilde{\xi}^{r+1})$, we have
\begin{eqnarray*}
\Phi_p(\xi^{r+1})&\leq&\Phi_p(\tilde{\xi}^{r+1})\\
 &=& \Phi_p(\xi^r)+\gamma\phi(\vx_r^*,\xi^r)+\frac{1}{2}\gamma^2\left.\frac{\partial^2\Phi_p((1-\gamma)\xi^r+\gamma\vx_r^*)}{\partial \gamma^2}\right|_{\gamma=\gamma'}\\
&<& \Phi_p(\xi^r)-\gamma a+\frac{1}{2}\gamma^2U,
\end{eqnarray*}
where $\gamma'$ is some value between 0 and 1. Consequently, we have
\begin{eqnarray*}
\Phi_p(\xi^{r+1})-\Phi_p(\xi^r)<\frac{1}{2}U\left(\gamma-\frac{a}{U}\right)^2-\frac{a^2}{2U}.
\end{eqnarray*}

If $\frac{a}{U}\leq 1\Leftrightarrow U\geq a$, choose $\gamma=\frac{a}{U}$, then we have
$$\Phi_p(\xi^{r+1})-\Phi_p(\xi^r)<-\frac{a^2}{2U}.$$
If $\frac{a}{U}> 1\Leftrightarrow U<a$, choose $\gamma=1$, then we have
$$\Phi_p(\xi^{r+1})-\Phi_p(\xi^r)<\frac{1}{2}U-a<0.$$
Both situations will lead to $\lim\limits_{r\rightarrow\infty} \Phi_p(\xi^r)=-\infty$, which  contradicts with $\Phi_p(\xi^r)>0$.

In summary, the assumption that $\xi^r$ does not converge to $\xi^*$ is not valid, and thus we prove that $\xi^r$ converges to $\xi^*$.
\end{proof}

\end{document}